\documentclass[11pt,letterpaper, english]{article}
\usepackage{amsmath,amsfonts,amssymb,amsthm}
\usepackage{microtype}
\usepackage[utf8]{inputenc}
\usepackage{float}
\usepackage{setspace}
\usepackage{rotating}
\usepackage{array}
\usepackage{multibib}
\usepackage[toc,page]{appendix}

\usepackage[flushleft]{threeparttable}
\usepackage{rotating}
\usepackage[subnum]{cases}
\usepackage{multirow, makecell}

\theoremstyle{plain}
\theoremstyle{definition}
\newtheorem{assumption}{Assumption}
\newtheorem{theorem}{Theorem}
\newtheorem{corollary}{Corollary}
\newtheorem{lemma}{Lemma}
\newtheorem{definition}{Definition}

\usepackage{xr}
\usepackage{color}
\usepackage[sectionbib,round]{natbib}
\usepackage{times}
\usepackage[small,compact]{titlesec} 
\usepackage{abstract}

\usepackage[hang,flushmargin]{footmisc}
\usepackage{fullpage} 
\usepackage{hyperref}
\usepackage[compact]{titlesec}
\usepackage[section]{placeins}
\usepackage{footnote}
\usepackage{epstopdf}
\usepackage{placeins}
\usepackage[toc,page]{appendix}
\usepackage{blkarray}
\usepackage{stackengine}
\usepackage{accents}
\usepackage{enumitem}
\usepackage{booktabs}  
\usepackage{nicefrac}
\usepackage{multirow}
\usepackage{subcaption}

\usepackage{longtable}
\usepackage{graphicx}
\usepackage{epstopdf}
\usepackage{bbm}
\usepackage{dsfont}
\usepackage{comment}
\usepackage{pgfplots}
\pgfplotsset{compat=1.15}

\DeclareMathOperator{\E}{\mathbb{E}}

\DeclareMathOperator{\X}{\mathbf{X}}

\DeclareMathOperator{\PP}{\mathcal{P}}
\DeclareMathOperator{\Q}{\mathbf{Q}}
\DeclareMathOperator{\FF}{\mathcal{F}}
\DeclareMathOperator{\J}{\mathbf{J}}
\DeclareMathOperator{\T}{\mathbf{T}}
\DeclareMathOperator{\Pii}{\mathbf{\Pi }}
\DeclareMathOperator{\LL}{\mathcal{L}}

\newcommand{\Lagr}{\mathcal{L}}

\newcommand{\squishlist}{
   \begin{list}{$\bullet$}
    { \setlength{\itemsep}{0pt} \setlength{\parsep}{1pt}
      \setlength{\topsep}{1pt} \setlength{\partopsep}{1pt}
      \setlength{\leftmargin}{1.5em} \setlength{\labelwidth}{1em}
      \setlength{\labelsep}{0.5em} } }

\newcommand{\squishlisttwo}{
   \begin{list}{$\bullet$}
    { \setlength{\itemsep}{0pt} \setlength{\parsep}{0pt}
      \setlength{\topsep}{0pt} \setlength{\partopsep}{0pt}
      \setlength{\leftmargin}{1em} \setlength{\labelwidth}{1.5em}
      \setlength{\labelsep}{0.5em} } }

\newcommand{\squishend}{
    \end{list}  }
\newcommand{\argmax}{\arg\!\max}

\newcommand{\hy}[1]{\noindent{\textcolor{blue}{\{{\bf Hema:}  #1\}}}}

\begin{document}
\title{A Recursive Partitioning Approach for Dynamic Discrete Choice Models in High-Dimensional Settings}
\author{
      Ebrahim Barzegary \thanks{We would like to thank Simha Mummalaneni and the participants of the 2023 ISMS Marketing Science Conference for feedback that has improved the paper. Please address all correspondence to: barzegary@essec.edu, hemay@uw.edu. }  \\
        ESSEC Business School\\
        \and
        Hema Yoganarasimhan\footnotemark[1]\\
        University of Washington\\
 }
\maketitle
\begin{abstract} 
Dynamic discrete choice (DDC) models are widely employed in marketing and economics to answer substantive and policy questions in settings where individuals' current choices have future implications. However, estimating these models is computationally intensive and/or infeasible in high-dimensional settings. Indeed, even specifying the structure for how the utilities/state transitions enter the agent's decision is challenging in high-dimensional settings when we have no guiding theory. We present a data-driven recursive partitioning algorithm that reduces the dimensionality of the high-dimensional state space by taking into account the variation in both choices and state transitions. The discretized state space from the first stage can then be used with any estimator of choice in the second stage to estimate the  DDC model. We present Monte-Carlo simulations on two general settings to highlight the performance of our approach -- (1) regenerative stationary setting (the canonical bus engine replacement problem), and (2) non-stationary finite horizon setting (durable goods adoption problem). Across both settings, our approach recovers the lower-dimensional representation of the state space and reduces estimation bias. Our approach can be applied to a variety of marketing problems where consumers make inter-temporal trade-offs, e.g., pricing and promotion planning, healthcare decisions, and subscription/renewal decisions.

\vspace{0.3in}

\noindent \textbf{Keywords:} Structural Models, Dynamic Discrete Choice Models, Machine Learning, Recursive Partitioning
\end{abstract}

\newpage
\section{Introduction}
\label{sec:intro}

Dynamic discrete choice models are extensively used in marketing and economics to study problems where agents make intertemporal trade-offs; see \citet{Aguirregabiria_Mira_2010, Aguirregabiria_2021} for detailed surveys. In marketing contexts, they have been used to model both firm and consumer behavior in a variety of settings, including:
\squishlist
\item Demand estimation for durable goods \citep{Song_Chintagunta_2003, Nair_2007, Esteban_Shum_2007, Gowrisankaran_Rysman_2012, Li_2019}
\item Demand estimation for storable goods \citep{Erdem_etal_2003, Hendel_Nevo_2006} 
\item Decision-making under uncertainty and consumer learning \citep{Erdem_Keane_1996}
\item Search and optimal stopping problems \citep{Kim_etal_2010, Yoganarasimhan_2013}
\squishend
Models in this class serve two primary purposes. First, they allow researchers to estimate the underlying structural parameters associated with an agent's utility function within an institutional context. Second, they serve as tools to evaluate new (counterfactual) policies, e.g., new pricing policies for durable goods \citep{Nair_2007}, positioning and product strategies by manufacturers and retailers \citep{Hitsch_2006, Ellickson_etal_2012, Melnikov_2013}, and consumer welfare \citep{Wang_2015}. 

In these models, an agent's decisions are dynamic, in the sense that the agent's current decisions affect their future utility stream through future states. As such, agents' current decisions depend on their expectations about the evolution of future states and their own future decisions in those states. Thus, estimating the underlying primitives of an agent's behavior in these settings requires the researcher to incorporate both the current utility an agent receives from a decision as well as her expected future stream of utilities conditional on this decision. Specifically, agents are assumed to maximize expected intertemporal payoffs in each time period. Thus, a key aspect of the standard solution concept is to calculate the value function, which is defined as the expected value of being in a specific state. Formally, this is equal to the discounted sum of the expected utility stream that an agent gets from making optimal choices starting with the current state. Typically, researchers calculate the value function by solving a dynamic programming problem using the Bellman equation \citep{ml_encyclopedia}. This is usually done using recursion when the problem is regenerative, or using backward induction when it is finite-horizon.

Three main issues make the specification and estimation of these models challenging. First, the well-known {\it curse of dimensionality} implies that it is exponentially costly to solve for the value function numerically since the dimensionality of the state space increases. Indeed, even specifying state transitions over an infinitely large state space is challenging. Thus, researchers have historically faced a trade-off between including a larger number of state variables (that can help improve the explanatory power of the model and give more realistic counterfactuals) vs. keeping the problem computationally tractable; see $\S$\ref{sec:related_lit} for detailed discussions. Second, even if the researcher has access to a high-dimensional set of variables, there is often no guiding theory on if and how these variables enter the agent's utility function and state transitions. Third, in most settings, the researcher is working with finite or limited samples of data. The combination of finite samples and high-dimensional state variables often implies that there may simply be no observations in many parts of the state space. Thus, the researcher has to reduce the dimensionality of the state space to ensure that there is sufficient data in all parts of the state space for inference.


To address these challenges, we propose a novel approach to model and accommodate high-dimensional state variables in dynamic discrete choice models. We specify the dynamic discrete choice model in terms of two sets of observable state variables -- (1) $\X$ -- a set of low dimensional state variables for which the researcher has some theory, i.e., we know that $\X$ affects the utility function/state transition (and may even know the functional form in which it does, though this is not necessary). (2) $\Q$ -- a set of high-dimensional state variables, which may or may not affect agents' utility function and/or state transition, but we do not know if and how they do so. For example, in the case of renewal decisions for a cloud subscription service like DropBox or AdobeCloud, $\X$ could consist of prices and product features offered, which we naturally expect to affect purchase decisions. And $\Q$ could consist of a large number of potential factors that may or may not affect demand, e.g., macroeconomic variables such as inflation, detailed data on consumers' product usage patterns (e.g., amount of time spent on the service, regularity of usage, diversity of features used), product descriptions on the product page, reviews and ratings from prominent websites, prices and features of related product categories, etc. 

We present an approach to reduce the dimensionality of $\Q$ using a data-driven discretization approach based on recursive partitioning. We assume that there exists a lower dimensional partitioning on $\Q$ ( $\Pii^*: \Q \to \PP$, where $\PP = \{1,\dots, k\}$, i.e., $k$ partitions) such that this partitioning is a sufficient statistic for $\Q$. We define this discretization of $\Q$ as \textit{perfect discretization}, i.e., a discretization where all the observations within a partition have similar decision and transition probabilities (conditional on $\X$ and the unobservable error terms). We then recast the dynamic discrete choice model from the high-dimensional $\Q$-space to the lower-dimensional $\PP$-space, and show the equivalence of these two models. 

The main challenge now is that we need to learn the discretization $\PP$ and estimate the structural parameters of interest, simultaneously. A naive approach here would be to simply use the log-likelihood of the data (as a function of the structural parameters) as the objective function for the recursive partitioning algorithm and generate the splits. However, this approach is practically infeasible for three reasons. First, under this approach, for each {\it potential} split, at each step of the partitioning algorithm, we need to solve the dynamic programming problem and estimate the structural parameters. Then we would compare likelihoods across all possible splits and then choose the split that maximizes the empirical likelihood of the data. This is prohibitively expensive when  $\Q$ is high-dimensional. Second, the likelihood function contains two sets of outcomes: (i) choice probabilities and (ii) state transition probabilities. Therefore, any data-driven approach to discretize the state space must account for both of these outcomes. This makes the splitting process more complex than the standard CART/causal tree, where there is only one set of estimands to be considered. Third, unlike a standard recursive partitioning algorithm, where we split on all the state variables, here we only split on $\Q$ since $\X$s are known (or assumed) to influence outcomes by definition. As such, we need an algorithm that only splits on a subset of state variables ($\Q$), but considers all the state variables  ($\{\X, \Q\}$) to estimate the choice probabilities and state transitions.

We design a recursive partitioning method that addresses all three of these problems. To address the first problem, we separate the discretization problem from the estimation procedure and suggest an objective function that is fully non-parametric. Thus, we do not need to solve the dynamic programming problem at each potential split; as a result, the evaluation of the objective function at each potential split is computationally cheap. To address the second problem, we split our objective function into two sub-objectives, whose relative importance can be learned from the data using model selection/hyperparameter tuning. To address the third problem, we customize the splitting procedure such that it only splits on $\Q$ and not $\X$. Finally, once we have a discretization, we can simply use the learned discretization ($\hat{\PP}$) in addition to $\X$ as the set of state variables in conjunction with any standard estimation method, e.g., Nested Fixed Point (NXFP), Two-step methods. 

Our dimensionality reduction method has several desirable properties that make it convenient to use and applicable to many settings. First, notice that we separate the estimation and discretization tasks and define a general discretization criterion that does not depend on the parametric assumptions of the estimation step. This property makes the algorithm robust to the parametric assumption of the estimation step. Furthermore, the discretization algorithm does not impose any limitation on the estimation method one can use in the second stage, which has advantages. In general, with high-dimensional state spaces, researchers often employ two-step methods that preclude the ability to perform full counterfactuals. However, in this case, researchers can also use full-solution methods such as nested fixed-point that allow for a full range of counterfactuals. Finally, our discretization algorithm has all the desirable properties of standard recursive partitioning-based algorithms, such as linear time complexity in the dimensionality of $\Q$, robustness to the scale of the $\Q$ variables and to the presence of irrelevant variables in $\Q$. Furthermore, the algorithm can be implemented in a parallelized way, making its execution fast and scalable.


We present two sets of simulations to highlight the empirical performance of our approach. In the first simulation study, we consider a high-dimensional version of the standard Rust bus engine problem \citep{Rust_1987}. This setting has been extensively used as a benchmark to compare the performance of newly proposed estimators/algorithms for single-agent dynamic discrete choice models \citep{Hotz_etal_1994, Aguirregabiria_Mira_2002, Arcidiacono_Miller_2011}. In the second simulation study, we use the durable good purchase problem as our setting \citep{Song_Chintagunta_2003, Nair_2007}. Durable goods adoption is a well-studied and important problem in marketing; as such, it is a good setting to explore the practical applicability of our approach. Across both studies, in addition to one low-dimensional state variable (mileage in the bus engine case and price in the durable goods adoption case), we allow for a 10-dimensional set of $\Q$ variables. We consider 12 different scenarios for how $\Q$ affects state transitions and flow utilities and how $\Q$ transitions from one state to another. In both simulations, across all the cases, we show that using our recursive partitioning approach to first partition the state space can help with parameter recovery and bias reduction compared to ignoring the high-dimensional state variables. 
Finally, in the Web Appendix, we also present simulations highlighting the importance of tuning hyper-parameters to use the data efficiently and also compare our approach to existing approaches for naive dimensionality reduction (e.g., PCA and $k$-means clustering).


Finally, we note that our approach is broadly applicable to many modern marketing problems. Today's firms typically have access to a large amount of information on detailed consumer behavior and product attributes (especially in digital environments), as well as the ability to deliver personalized marketing interventions that leverage customer heterogeneity. In many of these settings, consumers often make inter-temporal trade-offs, and there is significant heterogeneity in consumers' utilities as a function of their demographic and behavioral features. However, in many of these cases, we do not have extensive theory on which consumer attributes affect their utility and how, and the size of these features is often too large to rely on trial-and-error methods. Our estimation approach can help in these situations -- it can take a high-dimensional set of state variables as input and identify a lower-dimensional state-space that is relevant for consumer-decision-making, which can then be used to model consumers' utility functions and perform counterfactuals. Some natural examples of high-dimensional data settings where our approach can be helpful include: promotion planning and pricing (e.g., Instacart, Amazon, Walmart), software/services subscription and renewal (DropBox, Adobe, Netflix), and mobile health app notification and consumer outreach (e.g., FitBit, AppleWatch, MyChart). We provide an in-depth discussion of these potential applications in $\S$\ref{ssec:applications}.

\section{Related Literature}
\label{sec:related_lit}
First, our paper relates to the literature on the estimation of dynamic discrete choice models in high-dimensional settings. Many different methods have been proposed to address this problem. \citet{Hotz_Miller_1993} propose a two-step estimator, where the expected value function can be specified as an analytical function of the Conditional Choice Probability (CCP) of certain decisions and state transitions (under certain assumptions). \citet{Arcidiacono_Ellickson_2011} show that estimation can be further simplified in many settings by exploiting the ``finite-dependence property", by specifying the expected value functions as a function of terminal or renewal decisions. \citet{Chernozhukov_etal_2022} build on the two-step approach of \citet{Hotz_Miller_1993} and allow the researcher to estimate the CCPs and flow utilities flexibly using non- or semi-parametric methods and correct for the bias from the first-stage estimation in the second step. While these two-step methods can make estimation feasible in high-dimensional settings, they nevertheless have a few limitations. First, the researcher still needs to estimate non- or semi-parametric CCPs at different combinations of the state space, which is not hard to do accurately in high-dimensional settings with finite data. Second, these models require high-quality non-parametric models of state transitions because they rely on forward-simulation to estimate parameters. This is infeasible in most high-dimensional settings since it is hard to specify and estimate state transition models. In the reinforcement learning literature, models that rely on estimates of state transitions are often called `model-based' and are known to suffer from significant biases in high-dimensional settings unless the researcher has oracle knowledge of state transitions \citep{Zeng_etal_2023}. Third, to conduct counterfactual analyses, researchers usually need to solve the full model at least once using NFXP, which may again not be feasible in very high-dimensional settings. Finally, CCP methods do not provide insight into whether certain variables affect the flow utilities/state transitions and still require the researcher to make assumptions if/how these variables affect outcomes.

Other methods to simplify estimation in high-dimensional settings include state space reduction by relying on inclusive value sufficiency \citep{Gowrisankaran_Rysman_2012}, or approximation of the value function using parametric policy iteration \citep{Benitez_etal_2000} or sieve value function iteration \citep{Arcidiacono_etal_2012}. Typically, these methods work well when there exists a set of basis functions that provide a good approximation of the value function \citep{Rust_2000, Powell_2007}. A more recent body of literature has tried to use the advances in machine learning and methods such as neural networks to solve dynamic discrete choice problems. For example, \citet{Norets_2012} uses an artificial neural network to estimate the value function, taking state variables and parameters of interest as input. \cite{Su_Judd_2012} proposes yet another approach to solve the curse of dimensionality problem in dynamic discrete choice modeling, called MPEC. They formulate the dynamic discrete choice problem as a constrained maximization problem where the likelihood function is the maximization objective, and the bellman equations are the constraints. Nevertheless, the MPEC algorithm is still not applicable in high-dimensional settings without proper discretization of the state space because it requires us to estimate the value function at each point in the state space. Finally, all of these methods require the researcher to specify the functional form of the flow utility function and do not necessarily help with variable selection and/or state-space reduction when we have no a priori theory on which variables affect the problem and how.

Our algorithm adds to this literature by offering a method to break the curse of dimensionality through data-driven discretization of the state space. We show that state-space aggregation/discretization and parameter estimation can be treated as two separate tasks. Researchers can thus use our recursive partitioning algorithm to reduce the dimensionality of a high-dimensional state space $\Q$ to a categorical variable $\PP$ in the first stage and then use $\PP$ in addition to other independent variables $\X$ for parameter estimation in the second stage. In the second step, we can use any of the existing estimation algorithms, e.g., NXFP, two-step methods.



Our paper also contributes to the literature on estimation using recursive partitioning. \cite{Breiman_1984} proposed the original Classification and Regression Trees (CART) algorithm for prediction using recursive partitioning. Ensemble methods, such as random forests, combine several trees to produce better performance than a single tree \citep{Breiman_2001}. While recursive partitioning has traditionally been used for prediction tasks, recently, there has been interest in using it for causal estimation of heterogeneous treatment effects using methods such as causal trees and Generalized Random Forests \citep{Athey_Imbens_2016, Athey_etal_2019}. Our algorithm adds to this literature by proposing a novel use of recursive partitioning to estimate dynamic discrete choice models by formulating a non-parametric objective function that incorporates both choice probabilities and state transitions to reduce the dimensionality of the state-space.

Finally, our paper adds to the nascent but growing literature that combines structural models with machine learning. \citet{Jiang_etal_2023} develop a choice model for high-dimensional settings. \citet{Singh_etal_2023} exploit the permutation invariance property to propose nonparametric estimators to overcome the curse of dimensionality that is inherent in the non-parametric estimation of choice functions. \citet{Wei_Jiang_2022} discuss how neural networks can be used for parameter estimation in structural models. However, all these approaches study static choice models rather than dynamic choice models, which is our focus. 


\section{Problem Definition}
\label{sec:prob_def}
\subsection{Basic Set-up}
\label{ssec:basic_setup}
We consider the discrete choice problem from the perspective of a forward-looking single agent, denoted by $i \in \{1 \dots N\}$.  In every period $t$, the agent chooses between $j = 1\ldots J$. options. $i$'s decision in period $t$ is denoted by $d_{it}$, and $d_{it} = j$ indicates that the agent $i$ has chosen the option $j$ in period $t$. The agent's decision is not only a function of her utility in her current state ($s_{it}$) but also the expectation of her utility in all her future states given the decision $d_{it}$. We assume that the agent's state $s_{it}$ is composed of three sets of variables.
\begin{enumerate}
\itemsep-4px 
    \item A set of observable low-dimensional state variables $x_{it} \in \X$.
    \item A set of observable high-dimensional state variables $q_{it} \in \Q$.
    \item Unobservable state variable $\epsilon_{it}$, which is a $J \times 1$ vector each associated with one of the alternatives observed by the agent, but not by the researcher.
\end{enumerate}
$\X$ represents state variables for which we have some a priori theory, i.e., we know the parametric form in which they enter the utility function. $\Q$ denotes state variables for which we do not have a theory for if and how they influence users' decisions and state transitions. Note that it is possible to simply assume that $\X$ is null and simply put all the state variables in $\Q$ and learn how $\Q$ affects utilities and state transitions. That is, it is not necessary for the researcher to know/assume the functional form of any state variable, and the main variables of interest themselves could be high-dimensional.

In each period $t$, agent $i$ derives an instantaneous flow utility $u(s_{it},d_{it})$, which is a function of her decision $d_{it}$ and her state variables $s_{it}=\{x_{it},q_{it},\epsilon_{it}\}$. The per-period utility is additively separable as follows:
\begin{equation}
    u(s_{it},d_{it}=j) = \bar{u}(x_{it},q_{it},d_{it}=j;\theta_1) + \epsilon_{itj},
\end{equation}
\noindent where $\epsilon_{itj}$ is the error term associated with $j^{th}$ option at time $t$, $\bar{u}(x_{it},q_{it},d_{it}=j;\theta_1)$ is the deterministic part of the utility from making decision $j$ in state $x_{it}, q_{it}$, and $\theta_1$ is the set of structural parameters associated with the deterministic part of utility. The state $s_{it}$ transitions into new, but not necessarily different, values in each period following decision $d_{it}$. We make three standard assumptions on the state transition process -- first-order Markovian, conditional independence, and IID error terms. These assumptions imply that: (i) $\{s_{it},d_{it}\}$ are sufficient statistics for $s_{it+1}$, (ii) error terms are independent over time, and (iii) errors in the current period affect states tomorrow only through today's decisions. Thus, we have:
\begin{equation}  \Pr(x_{it+1},q_{it+1},\epsilon_{it+1}|x_{it},q_{it},\epsilon_{it},d_{it}) = \Pr(\epsilon_{it+1}) \Pr(x_{it+1},q_{it+1}| x_{it},q_{it},d_{it})
\end{equation}
We denote the state transition function $\Pr(x_{it+1},q_{it+1}| x_{it},q_{it},d_{it})$ as $g(x_{it+1},q_{it+1}| x_{it},q_{it},d_{it};\theta_2)$, where $\theta_2$ captures the parameters associated with state transition.\footnote{Both the utility function and state transition can also be estimated non-parametrically. In that case, $\theta_1$ and $\theta_2$ would simply denote the utility and state transition at a given combination of state variables.} 

Each period, the agent takes into account the current period payoff as well as how her decision today will affect the future, with the per-period discount factor given by $\beta$. She then chooses $ d_{it}$ to sequentially maximize the expected discounted sum of payoffs $ \E \left[ \sum_{\tau=t}^{\infty} \beta^{\tau} u(s_{i\tau},d_{i\tau}) \right]$. Our goal is to estimate the set of structural parameters $\theta =\{\theta_1,\theta_2\}$ that rationalizes the observed decisions and the states in the data, which are denoted by  $\{(x_{i1},q_{i1},d_{i1}), \dots, (x_{it},q_{it},d_{it}), \dots, (x_{iT},q_{iT},d_{iT})\}$ for agents $i \in \{1, \ldots, N\}$ for $T$ time periods.

\subsection{Challenges}
\label{subsec:curse_of_dimensionality}

The standard solution is to use a maximum likelihood method and estimate the set of parameters that maximizes the likelihood of observing the data. Given the first-order Markovian and conditional independence assumptions, we can write the likelihood function and estimate the structural parameters as follows:
\begin{align}
\label{eq:likelihood}
    \LL(\theta) = \sum_{i=1}^N \Bigg ( \sum_{t=1}^T \log \hat{p}(d_{it}|x_{it},q_{it};&\theta_1) + \sum_{t=2}^T \log \hat{g}(x_{it},q_{it}|x_{it-1},q_{it-1},d_{it-1};\theta_2)\Bigg )\\
\notag
    \theta^* &= \argmax_{\theta} \LL(\theta)
\end{align}
where $\hat{p}(.)$ is the predicted choice probabilities. 

However, there are three main challenges in estimating a model in this setting:
\squishlist
\item {\bf Theory:} First, the researcher may lack theory on how the high-dimensional variables $q$ enter the agents' utility function. For instance, if we consider the example of high-dimensional usage variables in a software subscription decision, we do not have much theoretical guidance on which product usage variables affect users' utility and how. As such, we cannot make parametric assumptions on the effect of $q$ on decisions and state transitions or hand-pick a subset of these variables to include in our model.

\item {\bf Data:} Second, in a high-dimensional setting, we may not have sufficient data in all areas of the state space to model the flow utility and the transitions to/from a given set of state variables. 

\item {\bf Estimation:} Finally, estimation of a discrete choice dynamic model in a high dimensional setting is often computationally infeasible and/or costly. There are two main difficulties in estimation in high-dimensional settings. First, standard approaches like \citet{Rust_1987}'s nested-fixed-point algorithm require us to calculate the value function at each combination of the state variables at each iteration of the estimation. In large state spaces, value function iteration/backward induction can be extremely slow or infeasible. One way to avoid this is to use two-step estimation methods \citep{Hotz_Miller_1993}, which avoid value function iteration. However, they nevertheless need non-parametric estimates of Conditional Choice Probabilities (CCPs) and state transition probabilities at each combination of state variables. While CCPs can be calculated flexibly when the number of decisions is relatively small (using semi-parametric approaches such as neural networks), it is not possible to specify and estimate state transition models in high-dimensional settings. Indeed, in the reinforcement learning literature, this challenge is often referred to as the lack of knowledge of the ``model of world'' and it is well-known that model-based roll-out approaches like two-step estimators are not feasible (or error-prone) without oracle knowledge of state transitions in high-dimensional settings \citep{Zeng_etal_2023}. 
\squishend
Thus, in a finite data high-dimensional setting where we lack guiding theory, it is not feasible to specify a utility function over $q$ and/or estimate a dynamic discrete choice model using conventional methods. Therefore, we need a data-driven approach that reduces the dimensionality of $\Q$ in an intelligent fashion. 




\subsection{Dimensionality Reduction using Discretization}
\label{ssec:dim_reduc_discret}

Dimensionality reduction is a popular procedure used in machine learning and statistics across multiple settings and applications. A common theme across these approaches is the notion that even in settings where the dimensionality of the problem space is very high, the space of \textit{relevant} dimensions is small. For example, the well-known matrix completion methods used in the recommendation systems literature rely on the {\it low-rank} assumption, i.e., they assume that the information in a high-dimensional user-item matrix can be represented by a smaller set of underlying factors\footnote{For example, in the movie recommendation example, these factors might be user preferences (e.g., liking action or comedy) and movie genres. Matrix completion exploits this low-rank structure: by finding the underlying factors that best explain the observed entries, we can then use these factors to predict the missing entries, effectively completing the matrix.}; see Chapter 10 of \citet{Wainwright_2019}. Similarly, in the reinforcement learning area, a natural approach to state space reduction is known as {\it state space abstraction}, where the goal is to pool together states that have the same or similar reward function (or flow utility) and/or policies \citep{Givan_etal_2003, Li_etal_2006}. This literature largely focuses on how to define different hierarchies of abstractions and bound the loss from the abstraction process from a theoretical perspective rather than on practical algorithms for obtaining abstractions. 

We build on these ideas and suggest a solution that recasts our problem by mapping $\Q$ to a lower-dimensional space through data-driven discretization such that $\Pii : \Q \to \PP$, where $\PP = \{1,\dots, k\}$. The goal is similar in spirit to that of the Classification and Regression Trees (CART) algorithm used for outcome prediction \citep{Breiman_1984} and the Causal Tree algorithm for estimation of conditional average treatment effects \citep{Athey_Imbens_2016}. These algorithms discretize the covariate space to minimize the heterogeneity of the statistic of interest within a partition. For example, the CART algorithm discretizes the covariate space into disjoint partitions such that observations in the same partition have similar outcome values/classes. Similarly, the Causal Tree algorithm discretizes the state space such that observations within the same partition have similar treatment effects. While the exact approaches from the static estimation of CART/causal tree cannot be directly translated to dynamic discrete choice models, we build on the notion of ``similarity within a partition" to adapt the high-level intuition from these models to our setting. Therefore, our first step is to outline the characteristics of a suitable discretization.

In $\S$\ref{sssec:good_discret} we formally define the term \textit{perfect discretization} as a discretization where the observations within a partition behave similarly. Then in, $\S$\ref{sssec:pi_formulation}, we show how the estimation problem can be reformulated for a lower-dimensional space by exploiting this definition.

\subsubsection{Properties of a Good Discretization}
\label{sssec:good_discret}
We now discuss some basic ideas that a good discretization should capture. First, some variables in $\Q$ may be irrelevant to our estimation procedure, i.e., they have no effect on flow utilities or state transitions. A good data-driven discretization should be able to neglect these variables and thus be robust to irrelevant variables. Second, our discretization approach has to be entirely non-parametric since we do not have any theory on how variables in $\Q$ affect agents' utilities and state transitions. Finally, the discretization should be generalizable, i.e., it should be valid outside the training data.
Formally, we define the term \textit{perfect discretization} as follows:


\begin{definition}
A discretization $\Pii^*:\Q \to \PP$ is perfect if all the points in the same partition have the same choice probabilities and incoming and outgoing transition probabilities. That is, for any two points $q,q'$ in a partition $\pi \in \PP$, we have:
\begin{numcases}{\label{eq:perfect_definition} \forall x,x' \in \X, q'' \in \Q, j \in \J:}
      \Pr(j|x,q) = \Pr(j|x,q')\\
      \Pr(x',q''|x,q,j)=\Pr(x',q''|x,q',j) \\
      \Pr(x,q|x',q'',j)=\Pr(x,q'|x',q'',j)
\end{numcases}
\end{definition}
The first equality ensures that the decision probabilities are similar for data points within a given partition $\pi \in \PP$. The second equality asserts the equality of transition probabilities {\it from} any two points $q, q'$ within the same partition $\pi \in \PP$. Finally, the last equality implies that the transition probabilities {\it to} any two points within a partition $\pi$ are equal. Together, these three equalities imply that all the observations within the same $\pi \in \PP$ are similar from both modeling and estimation perspective. Therefore, we do not need to model the heterogeneity within the partition $\pi$. Instead, modeling agents' behavior at the level of $\PP$ is sufficient. 

\subsubsection{Formulation of DDC in a discretized space}
\label{sssec:pi_formulation}
We now use the definition of perfect discretization and translate the DDC estimation from the $\Q$-space to the $\PP$-space. Because observations within each partition in a perfect discretization behave similarly, we can project the problem from the $\Q$-space to the $\PP$-space. That is, $\Pii^*$ is a sufficient statistic for the estimation of state transition and decision probabilities. We can therefore write the probabilities of choices and state transitions in the $\PP$-space as follows:
\begin{numcases}{\label{eq:definition_consequence} \forall x,x' \in \X, q'' \in \Q, j \in \J:}
      \Pr(j|x,q) = \Pr(j|x,\Pii^*(q))\\
      \Pr(x',q''|x,q,j)=\Pr(x',q''|x,\Pii^*(q),j) \\
      \Pr(x,q|x',q'',j) = \frac{\Pr(x,\Pii^*(q)|x',q'',j)}{N(x,\Pii^*(q))}
\end{numcases}
where $N(x,\Pii^*(q))$ is the number of observations in state $\{x,\Pii^*(q)\}$. 

These three equalities are the counterparts of the relationships shown in Equation \eqref{eq:perfect_definition} in the $\PP$-space. According to the first equality in Equation \eqref{eq:definition_consequence}, if the choice probabilities for the observations within a partition $\pi \in \PP$ are the same, then $\{\X,\PP\}$ is a sufficient statistic to capture the choice probabilities. The same argument is true for outgoing state transitions. If the probabilities of transition to other states from all points in a partition are similar, we can use the partition instead of points to specify transition (as shown in the second relationship in Equation \eqref{eq:definition_consequence}). Finally, when the probability of transition into all the points within a partition is similar, the probability of moving to a specific point is equal to the probability of transitioning into the partition divided by the number of observations within that partition. That is:
\begin{align*}
    \Pr(x,\Pii^*(q)|x',q',j) &=     
    \sum_{q'' \in \Pii^*(q)} \Pr(x,q''|x',q',j) \\
    &= N(x,\Pii^*(q)) \Pr(x,q|x',q',j),
\end{align*}
which gives us the third relationship in Equation \eqref{eq:definition_consequence}.

We now use the relationships in Equation \eqref{eq:definition_consequence} to reformulate the log-likelihood in Equation \eqref{eq:likelihood}. Given a perfect partitioning $\Pii^*$, the log-likelihood can be written as:
\begin{align}
\notag
    \LL(\theta,\Pii^*) &= \sum_{i=1}^N \sum_{t=1}^T \log p(d_{it}|x_{it},\Pii^*(q_{it});\theta_1,\Pii^*) \\
    &+ \sum_{i=1}^N \sum_{t=2}^T \log \frac{g(x_{it},\Pii^*(q_{it})|x_{it-1},\Pii^*(q_{it-1}),d_{it-1};\theta_2,\Pii^*)}{N(x_{it},\Pii^*(q_{it});\Pii^*)}   
    \label{eq:likelihood_pi*}
\end{align}
Note that the second term in the log-likelihood is obtained by combining the second and third terms in Equation \eqref{eq:definition_consequence} as: $ \Pr(x',q''|x,q,j) = \Pr(x',q''|x,\Pii^*(q),j) =\frac{\Pr(x',\Pii^*(q'')|x',\Pii^*(q),j)}{N(x,\Pii^*(q''))} $. Next, we can formulate the utility function in the $\Pii^*(q)$-space as follows:
\begin{equation}
        u(s_{it},d_{it}) = \bar{u}(x_{it},\Pii^*(q_{it}),d_{it};\theta_1,\Pii^*) + \epsilon_{itj}
\end{equation}
Finally, we can also write the value function and the choice-specific value function in terms of the discretized state space. Conceptually, once we have a perfect discretization $\Pii^*$, we can treat $\Pii^*(q)$ as a categorical variable in addition to $\X$, and ignore $q$. Note that under this new specification, the discretization $\Pii^*$ itself is a primitive that we need to learn, in addition to $\theta_1$, and $\theta_2$.

Note that once we have a discretization, i.e., a mapping from the $\Q$-space to $\hat{\PP}$-space, all the methods available for the estimation of dynamic discrete choice models (e.g., nested fixed point, two-step estimators) are directly applicable here, with $\{\X, \hat{\PP}\}$ as the state space. Thus, all the consistency and efficiency properties of the estimator used would directly translate to this setting, i.e., are valid conditional on the state variables, $\{\X, \hat{\PP}\}$. We provide a more detailed discussion of estimation and inference in $\S$\ref{ssec:est_summary}.

\subsection{Discussion}
\label{ssec:discussion}
We now provide a conceptual discussion of the model's interpretation when going from the basic set-up discussed in $\S$\ref{ssec:basic_setup} to the formulation in the lower-dimensional space discussed in $\S$\ref{sssec:pi_formulation}.

First, note that the conditions in Equation \eqref{eq:perfect_definition} imply that the agent/decision-maker knows which variables affect their decisions and state transitions and how. Indeed, when we rewrite the agent's problem in terms of the lower-dimensional $\PP$-space, the implicit assumption is that the agent is aware that these lower-dimensional states are all that matters for their decisions. However, the key point is that the researcher is unaware of this lower-dimensional partitioning and needs to learn the mapping from the full state space to the lower-dimensional space from the data. In that sense, the first-stage recursive partitioning is essentially recovering the relevant state space from the perspective of the decision-maker.

Second, in many settings, managers and researchers are interested in using estimation results to evaluate counterfactual scenarios. Therefore, we now discuss the counterfactual validity of the partitioning procedure. Mathematically, the key equations relevant to the validity of the partitions under the counterfactual scenarios are shown in Equation \eqref{eq:perfect_definition}. As long as we expect those to be valid in the counterfactual situation, the counterfactual results should be valid. Of course, there are some counterfactual scenarios where the conditions in Equation \eqref{eq:perfect_definition} either cannot be confirmed and/or are invalidated, as discussed below.
\squishlist 
\item First, as with all structural models, if the counterfactual scenarios are outside the bounds
of the current distribution of states observed in the estimation dataset, then we cannot guarantee that they would be valid without additional assumptions. Specifically, in regions outside the current dataset, the counterfactuals will need to rely on the assumption that users' flow utility and state transitions in those regions are similar to those in the states observed in the data. For example, suppose that our partitioning tells us that a state variable age can be partitioned into two groups such that users behave similarly within an age group -- (1) age between 18 and 35, and (2) age $>$ 35, but we only observe people till age 50 in the estimation data. Now suppose we have even older customers in a counterfactual setting (e.g., age $>$ 50). Then, by extrapolation, we would need to assume that the utility/transitions of these consumers are similar to the over-35 group observed in the estimation data.

\item Second, for the partitions to be valid in counterfactual scenarios, we need to assume that the partitions are policy invariant. This can be violated if the counterfactual scenario introduces changes to consumers' flow utilities or state transitions \textit{within} a partition. As before, consider a setting where consumers are partitioned into two age groups (say 18--35 years and $>$35 years) based on the estimation data. Now, suppose that the firm introduces a new promotion only for consumers aged 18--25 years. In that case, the first partition is no longer valid since consumers' flow utilities will vary depending on whether they are in the 18--25 age group or the 26--35 age group. Notice that in this example, the conditions in Equation \eqref{eq:perfect_definition} are violated under the counterfactual scenario. Thus, in this case, we will need to explicitly repartition the space to be consistent with the new promotional policy. Again, this is similar to a standard structural model, where the researcher will need to re-specify the flow utility in a counterfactual setting if factors that affect the flow utility change during the counterfactual.

\squishend
In summary, to the extent that the counterfactual scenario doesn't systematically change consumers' flow utilities and state transitions within a partition, the partitioning can be treated as a primitive and used in counterfactual analysis. If they directly affect the partitioning, they will need to be explicitly accounted for in the counterfactual. 

\section{Our Discretization Approach}
\label{sec:our_approach}
We now present our discretization algorithm. We first explain the recursive partitioning method for a general objective function in $\S$\ref{ssec:recur_part}. Next, in $\S$\ref{ssec:recur_part_DDC}, we present a recursive partitioning algorithm for the dynamic discrete choice model outlined above and discuss its properties. In $\S$\ref{ssec:cross_val}, we discuss the practical aspects of the algorithm, such as model selection and hyper-parameter tuning. In $\S$\ref{ssec:finite_horizon}, we discuss how to extend the discretization approach to finite horizon settings. Finally, in $\S$\ref{ssec:est_summary}, we summarize the estimation procedure and inference. 

\subsection{Discretization using Recursive Partitioning}
\label{ssec:recur_part}
\begin{figure}[!ht]
    \centering
    \includegraphics[width=16cm]{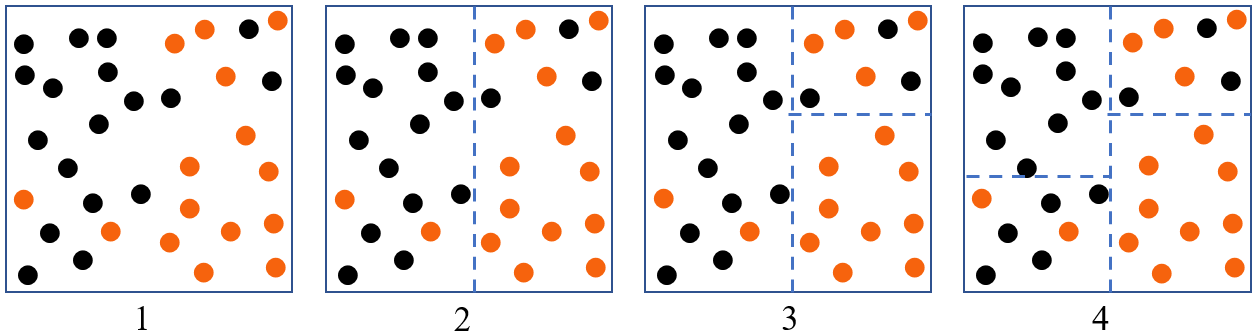}
    \caption{\label{figure:recursive} An example of recursive partitioning for a classification task with two explanatory variables and two outcome classes (denoted by orange and black dots).} 
\end{figure}
Recursive partitioning is a meta-algorithm for partitioning a covariate space into disjoint partitions to maximize an objective function. In each iteration, the algorithm uses an objective function as the criterion for selecting the next split among all the potential candidate splits. A split divides a partition into two, along one of the variables in the covariate space. The following pseudo-code presents a general recursive partitioning algorithm, where the goal is to maximize the objective function $\FF(\Pii)$.

\begin{itemize}
    \item Inputs: Objective function $\FF(\Pii)$ that takes a partitioning $\Pii$ as input and outputs a score.
    \item Initialize $\Pii_0$ as one partition equal to the full covariate space.
    \item Do the following until the stopping criterion is met, or $\FF(\Pii_{r}) = \FF(\Pii_{r-1})$
    \begin{itemize}
        \item For every $\pi \in \Pii_r$, every $q \in \Q$, and every value $v \in$ range($q$) in $\pi$
        \begin{itemize}[label={}]
            \item $\Delta(\pi,q,v) = \FF(\Pii_r + \{\pi, q, v\}) - \FF(\Pii_r)$
        \end{itemize}
        \item $\{\pi',q^*,v^*\} = \argmax_{\{\pi,q,v\}} \Delta(\pi,q,v)$
        \item $\Pii_{r+1} = \Pii_r + \{\pi',q^*,v^*\}$
    \end{itemize}
\end{itemize}

Figure \ref{figure:recursive} shows three iterations of recursive partitioning applied to a classification task. The partitioning $\Pii$ maps the two-dimensional covariate space into four different partitions. 

\subsection{Recursive Partitioning for Dynamic Discrete Choice Models}
\label{ssec:recur_part_DDC}
The goal of our discretization exercise is to estimate $\theta$ by maximizing the likelihood function in Equation \eqref{eq:likelihood_pi*}. To do so, we first need to find a perfect discretization, i.e., a discretization that satisfies Equation \eqref{eq:perfect_definition}. Thus, an intermediate goal is to find the partitioning $\Pii^*$. The key question then becomes what should be the objective function for the recursive partitioning algorithm that helps us achieve the two goals discussed above. A naive approach would be to simply use the log-likelihood shown in Equation \eqref{eq:likelihood_pi*}. However, this is problematic because of three reasons. 
\squishlist
\item First, this likelihood is a function of both $\Pii^*$ and $\theta$. Thus, a naive implementation of the recursive partitioning algorithm would require us to estimate the optimal $\theta$ in every iteration for a given $\Pii_r$. However, estimating $\theta$ is computationally expensive because we need to calculate the discounted future utility (or expected value function) associated with a given choice to estimate the primitives of agents' utility functions fully. Doing this at every {\it potential} split in each iteration of the algorithm is computationally expensive (and infeasible when the $\Q$-space is large). 
\item Second, the likelihood function contains two sets of outcomes: (i) choice probabilities and (ii) state transition probabilities. Any data-driven approach to discretize the state space must account for both of these outcomes. This makes the splitting process more complex than usual recursive partitioning algorithms such as CART/causal tree, where there is only one set of estimands to be considered. 
\item Third, unlike a standard recursive partitioning algorithm, where we split on all the state variables, here we only split on $\Q$ since $\X$s are known (or assumed) to influence outcomes by definition. As such, we need an algorithm that only splits on a subset of state variables ($\Q$), but considers all the state variables  ($\{\X, \Q\}$) to estimate the choice probabilities and state transitions; see Figure \ref{figure:q1x} for an example.
\squishend

\begin{figure}[!ht]
    \centering
    \includegraphics[width=60mm]{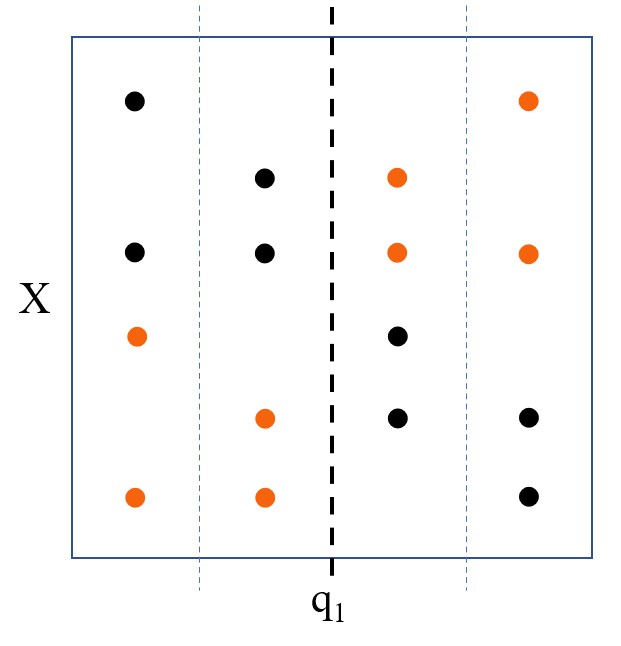}
    \caption{\label{figure:q1x} The dashed lines are candidate splits. Orange and black dots are observations associated with decisions 1 and 2 respectively. Note that agents' choices are a function of both $X$ and $q_1$.} 
\end{figure}

We design a recursive partitioning method that addresses all these three problems. To address the first problem, we separate the discretization problem from estimation and suggest an objective function that is fully non-parametric, i.e., is not dependent on $\theta$. Thus, its calculation is computationally cheap. To address the second problem, we split our objective function into two sub-objectives, whose relative importance can be learned from the data using model selection. Finally, to address the third problem, we customize the splitting procedure such that it only splits on $\Q$ and not $\X$. We discuss these ideas in detail below.

\subsubsection{Nonparametric Objective Function}
\label{sssec:nonpar_objective}
As we discussed in $\S$\ref{sssec:pi_formulation}, we can estimate the primitives of agents' utilities and state transitions by maximizing the log-likelihood shown in Equation \eqref{eq:likelihood_pi*}. This likelihood is a function of both $\theta$ and the discretization $\Pii^*$. Nevertheless, we do not need to solve the maximization problem on both dimensions simultaneously.  Conceptually, the discretization goal is to separate $\Q$ into buckets such that observations within the same bucket behave similarly. An important insight here is that we can achieve this objective without estimating $\theta$, by simply using the non-parametric estimates of choice probabilities and state transitions observed in the data. That is, we can re-formulate the objective function from Equation \eqref{eq:likelihood_pi*} in non-parametric terms. Our proposed objective function is thus a weighted sum of the nonparametric equivalents of the first and second components of Equation \eqref{eq:likelihood_pi*} as shown below.
\begin{equation}
\label{eq:ff_pi*}
    \FF(\Pii) = \FF_{dc}(\Pii) + \lambda \FF_{tr}(\Pii)
\end{equation}
where $\Pii$ is a partitioning function, and $\FF_{dc}(\Pii)$ and $\FF_{tr}(\Pii)$ are:
\begin{align}
\notag
    \FF_{dc}(\Pii) &= \sum_{i=1}^N \sum_{t=1}^T \log \frac{N(x_{it},\Pii(q_{it}),d_{it};\Pii)}{N(x_{it},\Pii(q_{it});\Pii)} \\
    \label{eq:ff_dc_pi*}
    &= \sum_{x \in \X} \sum_{\pi \in \Pii} \sum_{j \in \J} N(x,\pi,j; \Pii) \log \frac{N(x,\pi,j; \Pii)}{N(x,\pi; \Pii)}
\end{align}
\begin{align}
\notag
    \FF_{tr}(\Pii) &= \sum_{i=1}^N \sum_{t=2}^T \log \frac{N(x_{it},\Pii(q_{it}),x_{it-1},\Pii(q_{it-1}),d_{it-1};\Pii)}{N(x_{it-1},\Pii(q_{it-1}),d_{it-1};\Pii)N(x_{it},\Pii(q_{it}))} \\
    \label{eq:ff_tr_pi*}
     & = \sum_{x \in \X} \sum_{\pi \in \Pii}  \sum_{j \in \J} \sum_{x' \in \X} \sum_{\pi' \in \Pii} N(x,\pi,x',\pi',j; \Pii) \log \frac{N(x,\pi,x',\pi',j; \Pii)}{N(x,\pi; \Pii)N(x',\pi',j; \Pii)}
\end{align}
The function $N(.)$ counts the number of observations for a given condition. For example, $N(x,\pi,j)$ is the number of observations where the agent chose decision $j$ in state $\{x,\pi\}$ and $N(x,\pi,x',\pi',j)$ is the number of observations that chose decision $j$ in state $\{x',\pi'\}$, and transitioned to $\{x,\pi\}$. 

The weighting parameter $\lambda$ is a multiplier that specifies the relative importance of the state transition likelihood in comparison to decision likelihood in our algorithm. To make this parameter more intuitive and generalizable across different datasets, we decompose it as follows:
\begin{equation}
\label{eq:lambda_adj}
    \lambda = \lambda_{adj} \times \lambda_{rel}, \;\;\;
    \text{where } \lambda_{adj} = \frac{\FF_{dc}(\Pii_0)}{\FF_{tr}(\Pii_0)}.
\end{equation}
Here, $\Pii_0$ is the full covariate space of $\Q$ as one partition (i.e., the baseline case that ignores $\Q$). $\lambda_{rel}$ is a hyperparameter that captures the relative importance of state transition and decision likelihoods in our objective function and should be learned from the data using cross-validation. For example, $\lambda_{rel}=2$ implies that the recursive partitioning algorithm values one percentage lift in $\FF_{tr}$ twice as much a one percentage lift in $\FF_{dc}$ when selecting the next split. The optimal $\lambda_{rel}$ can vary with the application. In settings where there is more information in the state transition, the optimal $\lambda_{rel}$ will be higher whereas a smaller $\lambda_{rel}$ is better when the choice probabilities are more informative. Therefore, it is important to choose the right value of $\lambda_{rel}$ to prevent over-fitting and learn a good discretization.

A couple of additional points of note regarding our proposed objective function. First, the first lines in both Equations \eqref{eq:ff_dc_pi*} and \eqref{eq:ff_tr_pi*} are aggregating observations over individuals and time whereas the second lines are aggregating over all possible decisions and state transitions weighted by their occurrence. While the two representations are equivalent, we use the latter one in the rest of the paper since it is more convenient. Second, the objective function in Equation \eqref{eq:ff_pi*} is the non-parametric version the log-likelihood in Equation \eqref{eq:likelihood_pi*} (with the hyperparameter $\lambda$ added in for data-driven optimization). The terms $\frac{N(x_{it},\Pii(q_{it}),d_{it};\Pii)}{N(x_{it},\Pii(q_{it});\Pii)}$, and $\frac{N(x_{it},\Pii(q_{it}),x_{it-1},\Pii(q_{it-1}),d_{it-1};\Pii)}{N(x_{it-1},\Pii(q_{it-1}),d_{it-1};\Pii)}$ are the non-parametric counterparts of $\Pr(d_{it}|x_{it},\Pii(q_{it});\theta_1,\Pii)$ and $\Pr(x_{it},\Pii(q_{it})|x_{it-1},\Pii(q_{it-1}),d_{it-1};\theta_2,\Pii)$, respectively. Therefore, maximizing $\FF(\Pii)$ is equivalent to maximizing the original likelihood function. 

\subsubsection{Algorithm Properties}
\label{sssec:alg_prop}
 
We now establish two key properties of the recursive partitioning algorithm proposed here. First, as shown in Web Appendix \ref{subsec:likelihood_increase},  the non-parametric log-likelihood shown in Equation \eqref{eq:likelihood_pi*} is non-decreasing at each iteration of the algorithm. Formally, we have:
\begin{align*}
    \LL(\theta^*_r, \Pii_r) &\leq \LL(\theta^*_{r+1}, \Pii_{r+1})\\
    \text{where} &
    \begin{cases}
      \theta^*_r = \argmax_{\theta} \LL(\theta, \Pii_r)\\
      \theta^*_{r+1} = \argmax_{\theta} \LL(\theta, \Pii_{r+1})
    \end{cases}
\end{align*}
Second, as shown in Web Appendix\ref{subsec:likelihood_converge}, for $\lambda > 0$, the final discretization $\Pii^* = \argmax_{\Pii} \FF(\Pii)$ converges to a perfect discretization (under certain conditions). Together, these two properties ensure that: (a) the algorithm will increase the likelihood at each step and converge, and (b) upon convergence, it will achieve a perfect discretization. Of course, in practice, we may choose a $\lambda$ that does some bias-variance trade-off to avoid over-fitting and stop before we reach the perfect discretization to ensure that the discretization generalizes outside the training data.

Further, our proposed algorithm shares the desirable properties of other recursive partitioning-based algorithms. First, it has linear time complexity with respect to the dimensionality of $\Q$ \citep{Sani_etal_2018} -- if the number of dimensions in $\Q$ doubles, the algorithm's runtime at most doubles. Second, the proposed algorithm is robust to the scale of the state variables and only depends on the ordinality of the state variable. As such, any changes in the scale of variables in $\Q$ do not change the estimated partition. Finally, since the algorithm splits on a variable only if it increases the log-likelihood shown in Equation \eqref{eq:ff_pi*}, it is robust to the presence of irrelevant state variables. Together, these properties allow the researcher to include all potential observable variables that might affect the agents' decisions in $\Q$ without significantly increasing the compute cost. This is valuable in the current data-abundant era, where firms have massive amounts of user-level data but lack theoretical insight into the effect (if any) of these variables on users' decisions. Indeed, one of the advantages of the method is that it allows the researcher to make post-hoc inference or theory-discovery in a data-rich environment. 

Nevertheless, like other recursive partitioning algorithms, our algorithm does not guarantee consistency. This is because recursive partitioning algorithms are inherently unable to capture certain data patterns and are greedy by design; see Assumption 2 in Web Appendix$\S$\ref{app:proof_con} and \citet{Biau_etal_2008} for further details. Thus, it is not feasible for us to provide a formal consistency proof for the algorithm. Nevertheless, as we show in $\S$\ref{sec:simulation}, the algorithm performs well in a wide range of simulations. 



\subsection{Hyperparameter Optimization and Model Selection}
\label{ssec:cross_val}
Like other machine learning approaches, our recursive partitioning algorithm also needs to address the bias-variance trade-off. If we discretize $\Q$ into too many small partitions, the set of partitions (and the corresponding estimates of choice and state transition probabilities) will not generalize beyond the training data. Thus, we need a set of hyperparameters that constrain or penalize model complexity. We can then tune these hyperparameters using a validation procedure.\footnote{See \cite{Hastie_etal_2009} for a detailed discussion of the pros-cons of different validation procedures.}

\subsubsection{Set of Hyperparameters} 
We start with two hyperparameters that are commonly used in other recursive partitioning-based models: \textit{minimum number of observations} and \textit{minimum lift}.\footnote{See the hyperparameters in Random Forest \citep{Breiman_2001}, XGBoost \citep{Chen_Guestrin_2016}, and Generalized Random Forest \citep{Athey_etal_2019}.} The former stops the algorithm from making very small partitions by ruling out splits (at any given iteration) that produce partitions with observations fewer than the \textit{minimum number of observations}. The latter prevents over-fitting by stopping the partitioning process if the next split does not increase the objective function by the \textit{minimum lift}, which is defined as $\frac{\FF(\Pii_{r+1})-\FF(\Pii_{r})}{\FF(\Pii_{r})}$. In addition to these two standard hyperparameters, we also include another one: \textit{maximum number of partitions}, which stops the recursive partitioning after the algorithm has reached the \textit{maximum number of partitions}. This hyper-parameter not only controls over-fitting, but can also help with identification concerns because it allows the researchers to restrict the number of partitions, and hence the number of estimands. Together, these three hyperparameters ensure that the algorithm does not over-fit and produces a generalizable partition that is valid out-of-sample.




In addition to these three hyper-parameters that constrain partitioning, we have another key hyperparameter, $\lambda_{rel}$, that shapes the direction of partitioning. As discussed earlier, $\lambda_{rel}$ captures the relative importance of the two parts of the objective function when selecting the next split. If $\lambda_{rel}$ is set close to zero, the discretization procedure prioritizes splits that explain the choice probabilities. As $\lambda_{rel}$ increases, the recursive partitioning tends to choose splits that explain the variation in the state transition. $\lambda_{rel}$ can be either tuned using a validation procedure or set manually by the researcher based on their intuition or the requirements of the problem at hand. We present some simulations on the choice of $\lambda_{rel}$ in Web Appendix $\S$\ref{app:hyp_expt}. 


\subsubsection{Score function}

Let $\eta$ denote the set of all hyperparameters. To pick the right $\eta$ for a given application setting, we need a measure of performance at a given $\eta$. Then, we can consider different values of $\eta$ and select the one that maximizes out-of-sample performance (on the validation data). The model's performance is measured by calculating our objective function on a validation set using the training set's estimated values. Formally, our score function for a given set of hyperparameters is:
\begin{align}
\label{eq:score}
     score(\eta) &= \frac{1}{1+\lambda_{rel}}\Bigg ( \sum_{x \in \X} \sum_{\pi \in \Pii^*} \sum_{j \in \J} N^{val}(x,\pi,j; \Pii^*) \log \frac{N^{trn}(x,\pi,j; \Pii^*)}{N^{trn}(x,\pi; \Pii^*)}\\     \notag
     + \lambda_{rel}\lambda_{adj}^{val}& \sum_{x \in \X} \sum_{\pi \in \Pii^*}  \sum_{j \in \J} \sum_{x' \in \X} \sum_{\pi' \in \Pii^*} N^{val}(x,\pi,x',\pi',j; \Pii^*) \log \frac{N^{trn}(x,\pi,x',\pi',j; \Pii^*)}{N^{trn}(x,\pi; \Pii^*)N^{trn}(x',\pi',j; \Pii^*)} \Bigg )
\end{align}
where $\Pii^* = \argmax_{\Pii} \FF(\Pii;\eta)$ and $N^{trn}(.)$ and $N^{val}(.)$ are counting functions within the training and validation data, respectively. Notice that the  main differences between Equations \eqref{eq:ff_pi*} and \eqref{eq:score} is that here the model is {\it learned} on the training data, but evaluated on the validation data. As such, the weight terms ($N(\cdot)$) and $\lambda_{adj}$ are calculated on the validation dataset, while the probability terms are based on the training set. In addition, one minor difference is that this score is normalized by $\frac{1}{1+\lambda_{rel}}$. Without this normalization, any hyperparameter optimization procedure will tend to select larger $\lambda_{rel}$ since a larger $\lambda_{rel}$ leads to a higher score. 


In the validation procedure, we select the hyperparameter values that have the highest $score$ on the validation dataset. Note that having a separate validation set is not necessary. We can also use other hyper-parameter optimization techniques such as cross-validation (and this option is available in our accompanying code-base).

\subsubsection{Zero Probability Outcomes in Training Data}
\label{sssec:zero_prob}
A final implementation issue that we need to address is the presence of choices and state transitions in the validation set that have not been seen in the training set.\footnote{These kind of zero-probability occurrences are common in a finite sample setting as the dimensionality of the data increases.} For example, suppose that we have no observations where the agent chooses action $j$ in state $\{x,\pi\}$ in the training data, but there exists such an observation in the validation data. Then, the score in Equation \eqref{eq:score} becomes negative infinity. This problem makes the hyperparameter optimization very sensitive to outliers.


To solve this problem, we turn to the Natural Language Processing (NLP) literature, which also deals with high-dimensional data and has studied this problem extensively. Several smoothing techniques have been proposed to resolve this issue in NLP models \citep{Chen_1999}. One of the simplest smoothing methods that is used in practice and can be applied to our setting is additive smoothing \citep{Johnson_1932}, which assumes that we have seen all possible observations (i.e., choices and state-transitions) at least $\delta$ times, where usually $0 \leq \delta \leq 1$. This smoothing technique removes the possibility of zero probabilities

A typical value for $\delta$ in the NLP literature is 1; however, the optimal value for $\delta$ depends on the data-generating process. A low value for $\delta$ penalizes the model more heavily for potential outliers. It also prevents the recursive partitioning step from creating small partitions, since it increases the probability of having observations in the validation set that are not observed in the training set. On the other hand, a high value for $\delta$ may help with reducing over-fitting by adding more noise to the data. In our simulations, we set $\delta$ to $10^{-5}$, which is a relatively small value. However, the researcher can use a different number depending on their data and application setting.

\subsection{Extension to Finite Horizon Settings}
\label{ssec:finite_horizon}
So far, we focused on infinite-horizon stationary settings, which are used to model regenerative problems. We now discuss how to modify our approach to a finite-horizon setting, where time is also a state variable that impacts decision-making. Even in cases where time does not directly affect the flow utility, agents' decisions and the expected value of different decisions will change as they approach the end of the decision horizon. For instance, in a durable goods adoption problem \citet{Song_Chintagunta_2003, Nair_2007}, agents often wait to secure a better deal or higher quality. However, as time runs out, the utility of waiting diminishes, as failing to buy in later periods may result in no purchase at all. Thus, in finite-horizon settings, our dimensionality reduction algorithm needs to account for the influence of time on both purchase decisions and state transitions.


To that end, we redefine the concept of perfect discretization for finite horizon settings such that time is incorporated into the definition. By tailoring our algorithm to account for time transitions without explicitly treating time as an independent variable, we can incorporate time dynamics inefficiently. This approach increases the complexity of the algorithm linearly, by a factor of $T$, instead of quadratically, making it computationally feasible while still capturing the non-stationary nature of finite horizon decision-making.

\begin{definition}
A discretization $\Pii^*:\Q \to \PP$ is perfect in a finite horizon setting if all the points in the same partition at each period have the same choice probabilities and incoming and outgoing transition probabilities. That is, for any two points $q,q'$ in a partition $\pi \in \PP$ in time period $t \in \T$, we have:
\begin{numcases}{\label{eq:perfect_definition_finite_time} \forall x, x' \in \X, q'' \in \Q, j \in \J:}
      \Pr(j|x,q,t) = \Pr(j|x,q',t) \\
      \Pr(x',q''|x,q,j,t) = \Pr(x',q''|x,q',j,t) \\
      \Pr(x,q|x',q'',j,t) = \Pr(x,q'|x',q'',j,t)
\end{numcases}
\end{definition}

These equations are similar to those in Equation \eqref{eq:perfect_definition}, with the additional conditioning on time. It is easy to show that for estimation in the finite-horizon setting, we need to update Equations \eqref{eq:ff_dc_pi*} and \eqref{eq:ff_tr_pi*} as follows:
\begin{equation}
\label{eq:ff_dc_pi_finite}
    \FF_{dc}(\Pii) = \sum_{x \in \X} \sum_{\pi \in \Pii} \sum_{j \in \J} \sum_{t \in \T} N(x,\pi,j,t; \Pii) \log \frac{N(x,\pi,j,t; \Pii)}{N(x,\pi,t; \Pii)}
\end{equation}
\begin{equation}
\label{eq:ff_tr_pi_finite}
    \FF_{tr}(\Pii) = \sum_{x \in \X} \sum_{\pi \in \Pii} \sum_{j \in \J} \sum_{x' \in \X} \sum_{\pi' \in \Pii} \sum_{t \in \T} N(x,\pi,x',\pi',j,t; \Pii) \log \frac{N(x,\pi,x',\pi',j,t; \Pii)}{N(x,\pi,t; \Pii)N(x',\pi',j,t; \Pii)}
\end{equation}

By replacing the above equations into our scoring function in Equation \eqref{eq:score}, we can apply our algorithm to finite-horizon settings.  

\subsection{Summary of Estimation Approach and Inference}
\label{ssec:est_summary}
In summary, the estimation approach consists of two steps. In the first step, we employ the recursive partitioning algorithm described in $\S$\ref{ssec:recur_part}--$\S$\ref{ssec:cross_val} to reduce the dimensionality of the high-dimensional state variables $\Q$ to lower-dimensional $\PP$-space, and obtain a partition $\hat{\Pii}$. In the second step, we use $\X$ and the estimated $\hat{\Pii}$ as the set of observed state variables and estimate the dynamic discrete choice model. Note that at this stage, we can use any estimation procedure (including NXFP and two-step methods) to estimate the state-transition and flow utility parameters, and all the consistency and efficiency properties of these estimators are valid conditional of $\X$ and $\hat{\Pii}$.

 
However, an important caveat is that this is conditional on the estimated partition in the first step, $\hat{\Pii}$, and not on the true partition $\Pii^*$. To ensure that the estimated partitioning itself is consistent, we need to make additional assumptions on the nature of the patterns in the data, i.e., assume that there exists a true perfect discretization and that there are no data patterns that cannot be captured by recursive partitioning. The latter condition is a standard assumption in the literature on recursive partitioning methods; see \citet{Biau_etal_2008} for details. We formally state this as Assumption \ref{assumption:no_bad_pattern} along with an example of when this assumption can be violated in Web  Appendix $\S$\ref{app:proof_con}.

Further, a drawback of all recursive partitioning algorithms (e.g., CART) is that they are typically estimated using greedy algorithms and do not come with well-defined convergence rates and confidence intervals (or closed-form expressions for consistency and efficiency properties of the estimator).\footnote{The behavior of the recursive partition estimators continues to be an ongoing area of research; see \citet{Zheng_etal_2023}.} Our recursive partitioning algorithm shares the same drawbacks. Therefore, we cannot theoretically incorporate the rates of convergence/uncertainty from the first stage estimation of $\hat{\Pii}$ into the second stage estimation of flow utility or state transition parameters. As a result, it is not possible to provide a closed-form convergence rate or standard error for these parameters. Nevertheless, we provide bootstrap confidence intervals for the main parameters of interest in the Monte Carlo exercises shown in $\S$\ref{sec:simulation}.


\section{Monte Carlo Experiments}
\label{sec:simulation}
We now present two simulation studies to illustrate the performance of our algorithm. In the first simulation study, we use the canonical bus engine replacement problem introduced by \cite{Rust_1987} as the setting for our experiments. Rust's framework has been widely used as a benchmark to compare the performance of newly proposed estimators/algorithms for single-agent dynamic discrete choice models \citep{Hotz_etal_1994, Aguirregabiria_Mira_2002, Arcidiacono_Miller_2011}. In the second simulation study, we use the durable good purchase problem as our setting \citep{Song_Chintagunta_2003, Nair_2007}. Durable goods adoption is a well-studied and important problem in marketing; as such, it is a good setting to explore the practical applicability of our approach. 

\subsection{Infinite horizon case: Engine replacement problem}
\label{ssec:bus_engine}
In Rust's model, a single agent (Harold Zurcher) chooses whether to replace the engine of a bus or continue maintaining it in each period. The maintenance cost is linearly increasing in the mileage of the engine, while replacement constitutes a one-time lump-some cost. The intertemporal trade-off is as follows -- by replacing the bus engine, he pays a high replacement cost today and has lower maintenance costs in the future. If he chooses not to replace the bus engine, he avoids the replacement cost but will continue to pay a higher maintenance cost in the future. This is a stationary regenerative problem, i.e., time does not enter the state variables. At any point in time, the replacement decision resets the state space to a known and fixed value. 

\subsubsection{Data Generating Process}
We start with the basic version of the model without the high-dimensional state variables. Here, the per-period utility function of two choices is given by:
\begin{equation}
\label{eq:sim_utility_function}
\begin{aligned}
    u(x_{it},d_{it}=0) &= -c_m x_{it} + \epsilon_{i0t}\\
    u(x_{it},d_{it}=1) &= -c_r + \epsilon_{i1t},
\end{aligned}
\end{equation}
where $x_{it}$ is the mileage of bus $i$ at time $t$, $c_m$ is the per-mile maintenance cost, $c_r$ is the cost of replacing the engine, and $\{\epsilon_{i0t}, \epsilon_{i1t}\}$ are the error terms associated with the two choices. Next, we assume that the mileage increases by one unit in each period \footnote{This choice is just for the sake of increasing simplicity. However, our framework can easily handle stochastic state transitions as well.}. Formally: \begin{equation}
\label{eq:sim_state_transition}
\begin{aligned}
    &\text{if} \;\; d_{it} = 0,\;\;  \textrm{then} \;\;  x_{it+1} = x_{it} + 1 \;\;\textrm{if} \;\; x_{it} < 20, \;\;\textrm{else} \;\; x_{it+1} = x_{it} \\
    &\text{if} \;\; d_{it} = 1, \;\; \textrm{then} \;\;  x_{it+1} = 1 
\end{aligned}
\end{equation}
The maximum mileage is capped at 20, i.e., after the mileage hits 19, it continues to stay there.

We now extend the problem to incorporate a set of high-dimensional state variables $\Q$ that can affect the utility function and state transition. $\Q$ can include all the potential variables that can affect utilities and state transitions. For example, the mileage accrued may vary depending on the bus route, weather of the day, etc. Similarly, the replacement costs may vary by bus brands and/or economic conditions. A priori, it can be hard to identify which of these state variables and their combinations matter. Our method allows us to include all potential variables in the utility and state transition, and identify the partitions that matter.

In our simulations, we expand the utility function and state transition to include $\Q$ as follows:
\begin{equation}
\begin{aligned}
    u(x_{it},\Pii^*(q_{it}),d_{it}=0) &= c_m x_{it} + \epsilon_{i0t}\\
    u(x_{it},\Pii^*(q_{it}),d_{it}=1) &= f_{dc}(\Pii^*(q_{it})) +  \epsilon_{i1t} \\
\end{aligned}
\label{eq:extended_utility}
\end{equation}
\begin{equation}
\label{eq:extended_mileage_transition}
\begin{aligned}
    &\text{if} \;\; d_{it} = 0,\;\;  \textrm{then} \;\;  x_{it+1} = x_{it} + f_{tr}(\Pii^*(q_{it}))  \;\;\textrm{if} \;\; x_{it} < 20, \;\; \textrm{else} \;\; x_{it+1} = x_{it} \\
    &\text{if} \;\; d_{it} = 1 , \;\; \textrm{then} \;\;  x_{it+1} = f_{tr}(\Pii^*(q_{it}))
\end{aligned}
\end{equation}
where $q_{it}$ is the high-dimensional state variable for bus $i$ at time $t$, and $f_{dc}(\Pii^*(q_{it}))$ and $f_{tr}(\Pii^*(q_{it}))$ are functions that specify the effect of $q_{it}$ on the replacement cost and state transition, respectively. Note that we use $\Pii^*(q_{it})$ instead of $q_{it}$ since $\Pii^*(q_{it})$ is a perfect discretization (and conveys the same information) as $q_{it}$. Finally, we set $c_m = -0.2$ in both our simulation studies.\footnote{Note that the specific functional form is chosen for convenience, and the algorithm works even with fully non-parametric utilities within a partition and non-parametric state transitions across partitions.}

Next, we assume that $\Q$ consists 10 variable: $q^1, \dots, q^{10}$, where $q^i \in \{0,1,\dots,9\}$. However, only the first two variables affect the data-generating process, and the rest of them are irrelevant. In principle, our algorithm is robust to the inclusion of irrelevant state variables. Therefore, the extra eight variables in $\Q$ allow us to examine if this is indeed the case in practice. In our simulations, $q^1$ and $q^2$ partition $\Q$ into four regions $\{\pi_1,\pi_2,\pi_3,\pi_4\}$ such that all the observations within a partition have the same choice and state transition probabilities. The partitions are given by:
\begin{align}
\label{eq:pi_structure}
    \Pii^*(q) =     
    \begin{cases}
      \pi_1, & \text{if}\ q^1<5 \:\text{and}\: q^2<5 \\
      \pi_2, & \text{if}\ q^1<5 \:\text{and}\: q^2 \geq 5 \\
      \pi_3, & \text{if}\ q^1 \geq 5 \:\text{and}\: q^2<5 \\
      \pi_4, & \text{if}\ q^1 \geq 5 \:\text{and}\: q^2 \geq 5
    \end{cases}
\end{align}
The mileage transitions are as described in Equation \eqref{eq:extended_mileage_transition}. We also need to define the state transition in the $\Q$ space. Note that only the state transitions between $\pi$s matter, i.e., conditional on the partition $\pi$, the exact $q$ is not informative of utilities or state transitions random. We consider three different state transition models in our simulations:
\squishlist
    \item No transition: The agent remains within the same partition in each period: $\Pii^*(q_{it})=\Pii^*(q_{it+1})$.
      
    \item Random transition: Agents' transitions in the $\Q$-space are completely random. In each period, an agent randomly transitions from one partition to another such that: $\Pii^*(q_{it}) \perp \Pii^*(q_{it+1})$.
    
    \item Sparse transition: After each period, the agent remains in the same partition with probability $0.5$ and moves to the next partition with probability $0.5$. We choose the order of partitions as follows: $\pi_2$ is after $\pi_1$, $\pi_3$ is after $\pi_2$, $\pi_4$ is after $\pi_3$, and $\pi_1$ is after $\pi_4$.  
\squishend

Next, we consider two different options for $f_{tr}$ and $f_{dc}$ in Equations \eqref{eq:extended_utility} and \eqref{eq:extended_mileage_transition} each, as follows:
\begin{equation}
f_{tr}([\pi_1, \pi_2, \pi_3, \pi_4]) = 
    \begin{cases}
      [0,1,2,3] \text{ in the dissimilar mileage transition case}\\
      [1,1,1,1] \text{ in the similar mileage transition case}
    \end{cases}
\end{equation}
\begin{equation}
f_{dc}([\pi_1, \pi_2, \pi_3, \pi_4]) = 
    \begin{cases}
      [-7,-6,-5,-4] \text{ in the dissimilar replacement costs case}\\
      [-5,-5,-5,-5] \text{ in the  similar replacement costs case}
    \end{cases}
\end{equation}
Together, this gives us $3 \times 2 \times 2 = 12$ possible scenarios for the data-generating process. For example, one possible data-generating process could be $\{$No transition, Similar mileage transition, Dissimilar replacement cost$\}$, where $f_{tr}([\pi_1, \pi_2, \pi_3, \pi_4]) = [1,1,1,1]$ and $f_{dc}([\pi_1, \pi_2, \pi_3, \pi_4]) = [-7,-6,-5,-4]$. In this case, $\Q$ does not affect mileage transitions but only the flow utilities. On the other hand, if the data-generating process is $\{$Random transition, Dissimilar mileage transition, Similar replacement cost$\}$, then $\Q$ affects state transitions but not flow utilities. By considering all such possible scenarios, we are able to explore how the algorithm performs as the data-generating process changes.

We now simulate data and recover the structural parameters for each of these cases using the approach outlined in $\S$\ref{ssec:est_summary}. In the first step, we use our recursive partitioning procedure to generate the partitioning (as described in $\S$\ref{sec:our_approach}). While doing so, we set the hyperparameters as follows: {\it minimum lift} $= 10^{-10}$, {\it minimum number of observations} $= 1$, and $\lambda_{rel} = 1$. Further, for each case, we run the algorithm (and then perform the estimation) for four different values of {\it maximum partitions}: $1, 2, 4,$ or $6$. The single partition case is equivalent to ignoring the high-dimensional state variable $\Q$. Setting {\it maximum partitions} to 2 and 6 allows us to examine how our algorithm performs when we allow for under-discretization and over-discretization, respectively. Then, once we have a partition, we use the Nested Fixed Point algorithm by \citet{Rust_1987} to estimate the structural parameters for each case.   


\subsubsection{Results} 
For each of the 12 data generation processes, we run 100 simulations. In each simulation, we generate data for 400 buses for 100 time periods. For each simulated dataset, we first use our recursive partitioning algorithm to discretize the high-dimensional state space $\Q$ and obtain $\hat{\PP}$. While doing so, we consider four different versions of the recursive partitioning algorithm, where each version allows for a different value of {\it the maximum number of partitions} -- $\{1, 2, 4, 6\}$. Then, as discussed in $\S$\ref{ssec:est_summary}, we simply treat the set of estimated partitions in $\hat{\Pii}^*(q)$, i.e., $\hat{\PP}$ as an extra categorical variable in addition to $\X$ at the estimation stage, and estimate the structural utility parameters. The parameter estimates of $c_m$ from these simulations are presented in Table \ref{table:sim1_alt_result}. Next, in Table \ref{table:simple_simulation_res}, we focus only on two versions of the recursive partitioning model -- (1) where we allow for four partitions of $\Q$, and (2) one where we neglect $\Q$, which is similar to treating all the data as being in one partition, and present the replacement cost estimates. 

\begin{sidewaystable}[!htbp]
\centering
\footnotesize
\begin{tabular}{c|ccc|cccc}
\toprule
\textbf{Case} & \textbf{Effect on}  & \textbf{Effect on} & \textbf{Transition} & \multicolumn{4}{c}{\multirow{2}{*}{\textbf{Number of Allowed Partitions}}}\\
\textbf{No.} & \textbf{Replacement}  & \textbf{Mileage} &    \textbf{in $\Pii^*(\Q)$}  &  & \\
& \textbf{Cost}   & \textbf{Transition}   & \textbf{Space} & \textbf{1}& \textbf{2} & \textbf{4} & \textbf{6}\\
\midrule
Case 1 & Dissimilar & Dissimilar & No transition & -0.135 & -0.174 & -0.194 & -0.193 \\
 &  &  &  & (-0.144, -0.128) & (-0.18, -0.167) & (-0.205, -0.187) & (-0.205, -0.187) \\
Case 2 & Dissimilar & Dissimilar & Random transition & -0.149 & -0.186 & -0.200 & -0.199 \\
 &  &  &  & (-0.156, -0.144) & (-0.192, -0.18) & (-0.206, -0.193) & (-0.205, -0.193) \\
Case 3 & Dissimilar & Dissimilar & Sparse transition & -0.123 & -0.185 & -0.200 & -0.199 \\
 &  &  &  & (-0.128, -0.118) & (-0.196, -0.169) & (-0.209, -0.192) & (-0.208, -0.192) \\
Case 4 & Dissimilar & Similar & No transition & -0.164 & -0.191 & -0.199 & -0.198 \\
 &  &  &  & (-0.172, -0.157) & (-0.198, -0.184) & (-0.207, -0.19) & (-0.207, -0.189) \\
Case 5 & Dissimilar & Similar & Random transition & -0.172 & -0.193 & -0.200 & -0.200 \\
 &  &  &  & (-0.18, -0.166) & (-0.203, -0.182) & (-0.213, -0.19) & (-0.211, -0.189) \\
Case 6 & Dissimilar & Similar & Sparse transition & -0.170 & -0.099 & -0.200 & -0.200 \\
 &  &  &  & (-0.179, -0.163) & (-0.13, -0.067) & (-0.216, -0.191) & (-0.215, -0.191) \\
Case 7 & Similar & Dissimilar & No transition & -0.178 & -0.193 & -0.200 & -0.199 \\
 &  &  &  & (-0.186, -0.169) & (-0.2, -0.185) & (-0.206, -0.192) & (-0.206, -0.192) \\
Case 8 & Similar & Dissimilar & Random transition & -0.185 & -0.196 & -0.200 & -0.199 \\
 &  &  &  & (-0.191, -0.179) & (-0.204, -0.189) & (-0.207, -0.193) & (-0.207, -0.192) \\
Case 9 & Similar & Dissimilar & Sparse transition & -0.174 & -0.236 & -0.200 & -0.200 \\
 &  &  &  & (-0.184, -0.166) & (-0.252, -0.225) & (-0.213, -0.192) & (-0.213, -0.192) \\
Case 10 & Similar & Similar & No transition & -0.200 & -0.199 & -0.199 & -0.198 \\
 &  &  &  & (-0.207, -0.193) & (-0.207, -0.192) & (-0.207, -0.192) & (-0.206, -0.191) \\
Case 11 & Similar & Similar & Random transition & -0.200 & -0.200 & -0.199 & -0.198 \\
 &  &  &  & (-0.207, -0.193) & (-0.207, -0.192) & (-0.206, -0.19) & (-0.205, -0.19) \\
Case 12 & Similar & Similar & Sparse transition & -0.200 & -0.199 & -0.199 & -0.198 \\
 &  &  &  & (-0.206, -0.191) & (-0.206, -0.191) & (-0.205, -0.191) & (-0.205, -0.19) \\
\bottomrule
\end{tabular}

\caption{\label{table:sim1_alt_result} The estimated mileage maintenance cost, $c_m$, and the 96\% bootstrap confidence interval around the mean in each of the 12 Monte Carlo simulations. Each cell is a result of 100 rounds of simulation. Note that the case where the number of partitions is set to 1 is equivalent to completely neglecting $\Q$. }
\end{sidewaystable}
\begin{sidewaystable}[!htbp]
\centering
\footnotesize
\begin{tabular}{c|ccc|cccc|c}
\toprule
\textbf{Case} & \textbf{Effect on}  & \textbf{Effect on} & \textbf{Transition} & \multicolumn{4}{c}{\multirow{2}{*}{Incorporating discretized $\Q$}} & \multicolumn{1}{c}{\multirow{2}{*}{Neglecting $\Q$}}\\
\textbf{No.} & \textbf{Replacement}  & \textbf{Mileage} &    \textbf{in $\Pii^*(\Q)$}  &  & \\
 & \textbf{Cost}   & \textbf{Transition}   & \textbf{Sapce}         &                  $f_{dc}(\pi_1)$ &             $f_{dc}(\pi_2)$ &             $f_{dc}(\pi_3)$ &             $f_{dc}(\pi_4)$ &           {$c_r$}  \\
\midrule
Case 1 & Dissimilar & Dissimilar & No transition & -6.792 & -5.846 & -4.866 & -4.448 & -5.395 \\
 &  &  &  & (-7.15, -6.521) & (-6.14, -5.639) & (-5.165, -4.659) & (-4.684, -3.96) & (-5.675, -5.147) \\
Case 2 & Dissimilar & Dissimilar & Random transition & -7.011 & -6.002 & -4.997 & -4.001 & -4.965 \\
 &  &  &  & (-7.221, -6.814) & (-6.161, -5.834) & (-5.177, -4.844) & (-4.157, -3.836) & (-5.108, -4.829) \\
Case 3 & Dissimilar & Dissimilar & Sparse transition & -7.016 & -6.014 & -5.000 & -4.000 & -4.894 \\
 &  &  &  & (-7.251, -6.747) & (-6.189, -5.848) & (-5.153, -4.87) & (-4.212, -3.834) & (-5.065, -4.771) \\
Case 4 & Dissimilar & Similar & No transition & -6.986 & -5.987 & -4.987 & -3.989 & -4.687 \\
 &  &  &  & (-7.243, -6.717) & (-6.233, -5.774) & (-5.173, -4.832) & (-4.148, -3.863) & (-4.86, -4.557) \\
Case 5 & Dissimilar & Similar & Random transition & -7.033 & -6.009 & -5.018 & -4.012 & -4.670 \\
 &  &  &  & (-7.243, -6.8) & (-6.256, -5.766) & (-5.24, -4.825) & (-4.194, -3.833) & (-4.815, -4.545) \\
Case 6 & Dissimilar & Similar & Sparse transition & -7.043 & -6.009 & -5.004 & -4.001 & -4.725 \\
 &  &  &  & (-7.366, -6.817) & (-6.294, -5.803) & (-5.249, -4.811) & (-4.262, -3.843) & (-4.886, -4.592) \\
Case 7 & Similar & Dissimilar & No transition & -5.058 & -5.012 & -4.971 & -4.929 & -5.318 \\
 &  &  &  & (-5.223, -4.921) & (-5.115, -4.893) & (-5.088, -4.846) & (-5.059, -4.807) & (-5.511, -5.14) \\
Case 8 & Similar & Dissimilar & Random transition & -5.042 & -5.013 & -4.986 & -4.960 & -5.043 \\
 &  &  &  & (-5.206, -4.907) & (-5.188, -4.87) & (-5.136, -4.856) & (-5.104, -4.839) & (-5.181, -4.882) \\
Case 9 & Similar & Dissimilar & Sparse transition & -5.062 & -5.016 & -4.986 & -4.948 & -5.038 \\
 &  &  &  & (-5.272, -4.891) & (-5.205, -4.855) & (-5.179, -4.841) & (-5.114, -4.776) & (-5.226, -4.87) \\
Case 10 & Similar & Similar & No transition & -5.045 & -4.998 & -4.959 & -4.918 & -4.995 \\
 &  &  &  & (-5.206, -4.888) & (-5.135, -4.855) & (-5.11, -4.813) & (-5.062, -4.773) & (-5.123, -4.852) \\
Case 11 & Similar & Similar & Random transition & -5.036 & -4.998 & -4.974 & -4.934 & -5.000 \\
 &  &  &  & (-5.232, -4.892) & (-5.151, -4.878) & (-5.103, -4.843) & (-5.078, -4.767) & (-5.119, -4.869) \\
Case 12 & Similar & Similar & Sparse transition & -5.011 & -4.992 & -4.971 & -4.948 & -4.995 \\
 &  &  &  & (-5.142, -4.888) & (-5.116, -4.883) & (-5.102, -4.848) & (-5.085, -4.82) & (-5.118, -4.875) \\
\bottomrule
\end{tabular}
\caption{\label{table:simple_simulation_res} The estimated replacement cost and their 96\% bootstrap confidence interval around the mean in each of the 12 Monte Carlo simulations. Each row is a result of 100 rounds of simulation. }
\end{sidewaystable}

As we can see from both tables, neglecting $\Q$ often leads to biased estimates; see column 5 in Table \ref{table:sim1_alt_result} and the last column in Table \ref{table:simple_simulation_res}. This is true even when $q_{it}$ does not directly affect the flow utility and there is no serial correlation in $q$ over time, e.g., Case 8 in Table \ref{table:sim1_alt_result}. In this case, the bias arises because $\Q$ affects the state transitions through $f_{tr}$ (which in turn affects the expected future value function). Thus, when we neglect $\Q$ in the estimation procedure, it is captured in the error term, which in turn violates the conditional independence assumption that $\epsilon_{it+1}$ and $x_{it+1}$ are independent of $\epsilon_{it}$. Further, the bias seems to be worse when $\Q$ affects both flow utility and state transition, especially when $\Q$ itself has state dependence in how it transitions into the next period (e.g., case 3). Also, as expected, in cases where $\Q$ does not affect the utility function or the state transition (cases 10, 11, and 12 in both tables), neglecting $\Q$ does not bias the parameter estimates. This is because $\Q$ is an irrelevant state variable in these cases. 

Finally, columns 6 and 8 in Table \ref{table:sim1_alt_result} highlight the impact of under- and over-discretization. As expected, over-discretization does not introduce any bias in the estimation procedure since it does not violate any estimation assumption. However, it does lead to marginally higher variance (compared to column 7) since we estimate more parameters when we over-discretize. However, under-discretization does lead to bias; though, by and large, the bias seems to be lower than completely ignoring $\Q$. In sum, these simulations show that our recursive partitioning approach can help reduce the dimensionality of the dynamic discrete choice problem and help with correcting for bias and recovering true primitive parameters.

\subsection{Non-stationary Setting: Durable Goods Adoption}
\label{ssec:finite_horizon_simulations}

In this section, we apply the non-stationary version of our recursive partitioning algorithm (RePaD), from $\S$\ref{ssec:finite_horizon} to the durable goods adoption problem. Durable goods adoption is an important and well-studied problem in the marketing literature; see \cite{Song_Chintagunta_2003} and \cite{Nair_2007}, who apply the dynamic discrete choice framework to this problem in the case of digital cameras and video-game consoles, respectively. In non-stationary settings, the Markov process is not regenerative. This could happen because the time horizon is finite and/or because time, as a state variable, directly affects the flow utility or evolution of certain state variables. For example, in the durable goods adoption case, prices usually decrease over the lifetime of the product, and/or the effective quality increases. Alternatively, the product may have a well-defined life-cycle after which consumers no longer purchase the product.




\subsubsection{High-dimensional version of durable goods adoption problem}

To simulate the finite horizon durable goods adoption problem, we adopt a simplified version of the model from \cite{Song_Chintagunta_2003}. In this model, consumers decide when to make a one-and-done purchase of a durable product over a finite time horizon $\T$. Each period, they evaluate the utility of purchasing immediately versus waiting for future price drops. The inter-temporal trade-off is as follows -- purchasing today allows the consumer to start enjoying the product right away (higher consumption utility), but the purchase itself is costlier now than in the future (since prices go down in expectation). On the other hand, waiting yields zero consumption utility in the current period but allows you to buy the product at a lower price in the future.

For our experiments, we consider a high-dimensional version of this problem, where there are two sets of state variables -- (1) a low-dimensional state variable that we know affects utilities, price, $p_{it}$, and (2) a set of high-dimensional state variables, $Q_{it}$ that capture a variety of information about the consumer, market, and product that can potentially affect consumers' decisions. For example, this could include consumers' demographics (e.g., age, income, geographic location, gender, job description), product features (e.g., memory, product size, ease of use, camera resolution), and market and macro-economic variables (e.g., availability of competing products, their features, and prices, and other metrics on economic conditions). In settings where the consumer already has a base/older version of the product (e.g., say an older version of an iPhone), $Q_{it}$ could also capture a number of usage/behavioral features (e.g., how much the user uses the product, number of apps installed, regularity of usage, payments made in different apps). As such, depending on the setting and data availability, $Q_{it}$ can be very high-dimensional. Further, in many settings, it is unclear exactly which of these state variables affect consumers' utilities and state transitions and how. 

One option is to ignore these high-dimensional state variables and/or include a small subset of them and assume a specific parametric form for how they transition and enter consumers' utilities. However, doing so would not allow us to capture the full extent of consumer-level heterogeneity or the full scope of market features, which in turn can lead to imperfect substantive insights and/or incorrect counterfactual analysis. For example, \citep{Ryan_Tucker_2012} study the adoption of a video-calling technology in a firm, where the researchers have access to a large set of user features (with a state space larger than $10^{602}$). They find that user features have significant ability to explain users' adoption decisions and that ignoring them can lead to substantially different insights. Therefore, it is important to be able to accommodate the high-dimensional state variables and flexibly model consumer-level heterogeneity, which our approach allows us to do.

\subsubsection{Data generating process} 
We now provide the details of the data-generating process used in the simulations. The utility of purchasing at time $t$ is given by:
\begin{equation}
    u(p_{it},q_{it},d_{it}=1) = \frac{f_{dc}(q_{it})}{1 - \beta} - \alpha p_{it} + \epsilon_{it},
\end{equation}
where $f_{dc}(q_{it})$ is a function representing the effect of the high-dimensional state-variables ($q_{it}$) on utility, $\beta$ is the discount factor, $p_{it}$ is the price of the product at time $t$ for consumer $i$, and $\epsilon_{it}$ is an i.i.d. Type I Extreme Value error term. Further, the price dynamics are modeled as follows:
\begin{equation}
    p_{it+1} = p_{it} - f_{tr}(q_{it}),
\end{equation}
where $f_{tr}(q_{it})$ is a function of the high-dimensional state variables $q_{it}$, determining the rate of price decay for consumer $i$ at time $t$. Prices are discretized in increments of 1 for simplicity.

The structure of $q_{it}$ is similar to the first simulation study, consisting of a 10-dimensional set of state variables, which are discretized into four regions $\{\pi_1, \pi_2, \pi_3, \pi_4 \}$ based on the first two variables, as described in Equation \eqref{eq:pi_structure}. Further, the $q_{it}s$ transition in the same way as described in the first simulation, i.e., we consider three transition models (no transition, random, and sparse transition).

Next, we consider two scenarios each for the functions $f_{tr}$ and $f_{dc}$:
\begin{equation}
f_{tr}([\pi_1, \pi_2, \pi_3, \pi_4]) = 
    \begin{cases}
      [3,2,1,0] \text{ in the dissimilar price transition case}\\
      [1,1,1,1] \text{ in the similar price transition case}
    \end{cases}
\end{equation}
\begin{equation}
f_{dc}([\pi_1, \pi_2, \pi_3, \pi_4]) = 
    \begin{cases}
      [19.4,19.6,19.8,20] \text{ in the dissimilar purchase utility case}\\
      [20,20,20,20] \text{ in the  similar purchase utility case}
    \end{cases}
\end{equation}
Overall, as before, this gives us 12 scenarios for the simulations. Next, we set the starting price at $t=1$ as \$400, the price coefficient in the utility function is set to $0.5$, and the number of periods is $\T = 30$. An agent who makes the purchase decision exits the market. All the other details of the data-generating process and discretization are similar to the first simulation.

Finally, because we now have a finite horizon problem, we use backward induction as the solution concept for obtaining value functions at each step instead of value function iteration. Backward induction leverages the finite time horizon to solve the Bellman equation iteratively, starting from the last period $t = \T$ and working backward to the first period $t = 1$, for a given value of parameter estimates.

\begin{sidewaystable}[!htbp]
\centering
\footnotesize
\begin{tabular}{c|ccc|cccc}
\toprule
\textbf{Case} & \textbf{Effect on}  & \textbf{Effect on} & \textbf{Transition} & \multicolumn{4}{c}{\multirow{2}{*}{\textbf{Number of Allowed Partitions}}}\\
\textbf{No.} & \textbf{Purchase}  & \textbf{Price} &    \textbf{in $\Pii^*(\Q)$}  &  & \\
& \textbf{Utility}   & \textbf{Transition}   & \textbf{Space} & \textbf{1}& \textbf{2} & \textbf{4} & \textbf{6}\\
\midrule
Case 1 & Dissimilar & Dissimilar & No transition & -0.146 & -0.372 & -0.499 & -0.501 \\
 &  &  &  & (-0.151, -0.139) & (-0.383, -0.361) & (-0.513, -0.484) & (-0.515, -0.487) \\
Case 2 & Dissimilar & Dissimilar & Random transition & -0.447 & -0.451 & -0.498 & -0.504 \\
 &  &  &  & (-0.464, -0.43) & (-0.465, -0.437) & (-0.514, -0.482) & (-0.521, -0.485) \\
Case 3 & Dissimilar & Dissimilar & Sparse transition & -0.358 & -0.408 & -0.499 & -0.504 \\
 &  &  &  & (-0.368, -0.347) & (-0.418, -0.398) & (-0.513, -0.484) & (-0.528, -0.484) \\
Case 4 & Dissimilar & Similar & No transition & -0.295 & -0.443 & -0.499 & -0.502 \\
 &  &  &  & (-0.306, -0.284) & (-0.455, -0.429) & (-0.513, -0.484) & (-0.52, -0.484) \\
Case 5 & Dissimilar & Similar & Random transition & -0.383 & -0.440 & -0.499 & -0.509 \\
 &  &  &  & (-0.397, -0.369) & (-0.455, -0.424) & (-0.521, -0.479) & (-0.531, -0.491) \\
Case 6 & Dissimilar & Similar & Sparse transition & -0.380 & -0.480 & -0.500 & -0.504 \\
 &  &  &  & (-0.395, -0.368) & (-0.5, -0.462) & (-0.522, -0.48) & (-0.529, -0.482) \\
Case 7 & Similar & Dissimilar & No transition & -0.320 & -0.484 & -0.500 & -0.502 \\
 &  &  &  & (-0.338, -0.305) & (-0.502, -0.467) & (-0.516, -0.482) & (-0.519, -0.486) \\
Case 8 & Similar & Dissimilar & Random transition & -0.538 & -0.506 & -0.497 & -0.497 \\
 &  &  &  & (-0.56, -0.519) & (-0.522, -0.489) & (-0.512, -0.481) & (-0.519, -0.48) \\
Case 9 & Similar & Dissimilar & Sparse transition & -0.578 & -0.490 & -0.497 & -0.499 \\
 &  &  &  & (-0.597, -0.554) & (-0.504, -0.475) & (-0.514, -0.483) & (-0.516, -0.482) \\
Case 10 & Similar & Similar & No transition & -0.500 & -0.500 & -0.499 & -0.501 \\
 &  &  &  & (-0.527, -0.48) & (-0.526, -0.479) & (-0.523, -0.477) & (-0.519, -0.478) \\
Case 11 & Similar & Similar & Random transition & -0.500 & -0.499 & -0.498 & -0.500 \\
 &  &  &  & (-0.518, -0.48) & (-0.517, -0.48) & (-0.517, -0.477) & (-0.517, -0.476) \\
Case 12 & Similar & Similar & Sparse transition & -0.498 & -0.498 & -0.497 & -0.499 \\
 &  &  &  & (-0.517, -0.476) & (-0.516, -0.475) & (-0.515, -0.475) & (-0.519, -0.482) \\
\bottomrule
\end{tabular}

\caption{\label{table:sim1_alt_finite_result} The estimated price coefficient, and the 96\% bootstrap confidence interval around the mean in each of the 12 Monte Carlo simulations. Each cell is a result of 100 rounds of simulation. Note that the case where the number of partitions is set to 1 is equivalent to completely neglecting $\Q$. }
\end{sidewaystable}
\begin{sidewaystable}[!htbp]
\centering
\footnotesize
\begin{tabular}{c|ccc|cccc|c}
\toprule
\textbf{Case} & \textbf{Effect on}  & \textbf{Effect on} & \textbf{Transition} & \multicolumn{4}{c}{\multirow{2}{*}{Incorporating discretized $\Q$}} & \multicolumn{1}{c}{\multirow{2}{*}{Neglecting $\Q$}}\\
\textbf{No.} & \textbf{Purchase}  & \textbf{Price} &    \textbf{in $\Pii^*(\Q)$}  &  & \\
 & \textbf{Utiliy}   & \textbf{Transition}   & \textbf{Sapce}         &                  $f_{dc}(\pi_1)$ &             $f_{dc}(\pi_2)$ &             $f_{dc}(\pi_3)$ &             $f_{dc}(\pi_4)$ &           {$constant$}\\
\midrule
Case 1 & Dissimilar & Dissimilar & No transition & 19.36 & 19.56 & 19.76 & 19.96 & 5.68 \\
 &  &  &  & (18.79, 19.9) & (18.98, 20.12) & (19.17, 20.32) & (19.37, 20.53) & (5.42, 5.91) \\
Case 2 & Dissimilar & Dissimilar & Random transition & 19.33 & 19.54 & 19.74 & 19.94 & 17.82 \\
 &  &  &  & (18.72, 19.93) & (18.89, 20.14) & (19.09, 20.34) & (19.29, 20.54) & (17.11, 18.5) \\
Case 3 & Dissimilar & Dissimilar & Sparse transition & 19.36 & 19.56 & 19.76 & 19.96 & 14.19 \\
 &  &  &  & (18.78, 19.87) & (18.98, 20.11) & (19.18, 20.31) & (19.37, 20.51) & (13.76, 14.6) \\
Case 4 & Dissimilar & Similar & No transition & 19.36 & 19.56 & 19.75 & 19.95 & 11.53 \\
 &  &  &  & (18.76, 19.93) & (18.96, 20.14) & (19.16, 20.33) & (19.35, 20.54) & (11.1, 11.97) \\
Case 5 & Dissimilar & Similar & Random transition & 19.34 & 19.55 & 19.75 & 19.95 & 15.15 \\
 &  &  &  & (18.55, 20.2) & (18.74, 20.42) & (18.94, 20.62) & (19.13, 20.82) & (14.6, 15.71) \\
Case 6 & Dissimilar & Similar & Sparse transition & 19.39 & 19.59 & 19.79 & 19.99 & 15.02 \\
 &  &  &  & (18.6, 20.26) & (18.79, 20.47) & (18.98, 20.67) & (19.18, 20.88) & (14.57, 15.63) \\
Case 7 & Similar & Dissimilar & No transition & 19.98 & 19.99 & 19.99 & 20.00 & 12.89 \\
 &  &  &  & (19.28, 20.63) & (19.29, 20.64) & (19.29, 20.65) & (19.3, 20.66) & (12.31, 13.65) \\
Case 8 & Similar & Dissimilar & Random transition & 19.89 & 19.89 & 19.90 & 19.90 & 21.58 \\
 &  &  &  & (19.22, 20.48) & (19.22, 20.49) & (19.22, 20.49) & (19.22, 20.49) & (20.79, 22.46) \\
Case 9 & Similar & Dissimilar & Sparse transition & 19.89 & 19.89 & 19.90 & 19.90 & 23.21 \\
 &  &  &  & (19.29, 20.56) & (19.29, 20.56) & (19.3, 20.56) & (19.3, 20.57) & (22.22, 23.97) \\
Case 10 & Similar & Similar & No transition & 19.94 & 19.94 & 19.95 & 19.95 & 20.02 \\
 &  &  &  & (19.06, 20.93) & (19.06, 20.93) & (19.06, 20.94) & (19.07, 20.94) & (19.19, 21.09) \\
Case 11 & Similar & Similar & Random transition & 19.93 & 19.93 & 19.93 & 19.93 & 20.00 \\
 &  &  &  & (19.07, 20.67) & (19.07, 20.67) & (19.07, 20.67) & (19.07, 20.67) & (19.2, 20.72) \\
Case 12 & Similar & Similar & Sparse transition & 19.87 & 19.87 & 19.87 & 19.87 & 19.94 \\
 &  &  &  & (18.99, 20.61) & (18.99, 20.62) & (18.99, 20.62) & (18.99, 20.62) & (19.03, 20.67) \\
\bottomrule
\end{tabular}
\caption{\label{table:simple_simulation_finite_res} The estimated purchase utility and their 96\% bootstrap confidence interval around the mean in each of the 12 Monte Carlo simulations. Each row is a result of 100 rounds of simulation. }
\end{sidewaystable}

\subsubsection{Results}
As before, we simulate data and recover the utility parameters for each of these cases. Specifically, in the first step, we use our recursive partitioning procedure to generate the partitioning. The hyper-parameters are fixed to the same values as the first simulation study. Once we have the partitions, we use the backward induction Nested Fixed Point algorithm by \citet{Rust_1987} in combination with backward induction to estimate the structural parameters for each case. Further, as before, we have 12 data generation scenarios, and we run 100 simulations for each scenario. Each simulation consists of data from 10000 consumers. Table  \ref{table:sim1_alt_finite_result} shows the estimates for the price coefficient for the 12 scenarios for four different partitioning scenarios. Table  \ref{table:simple_simulation_finite_res} shows the estimates of $f_{dc}$ in all the 12 scenarios when we set the number of partitions to 4.

We see that our approach can recover the true parameter estimates without significant bias, even when we allow for over-partitioning. Further, neglecting $\Q$ leads to biased estimation, even in cases where the high-dimensional state variables only affect the state transition and not the purchase utility. This bias occurs across all scenarios except the last three, where $q_{it}$ does not affect state transitions or flow utilities. As before, the extent of bias varies depending on the degree of interaction between $\Q$ and the functions $f_{tr}$ and $f_{dc}$. In scenarios with dissimilar price transitions and purchase utilities, the omission of $\Q$ results in greater bias compared to scenarios with similar transitions or utilities. In summary, our simulations suggest that in high-dimensional settings where consumer heterogeneity is of importance, our method can help accurately capture and model this heterogeneity.

\subsection{Other Simulations and Checks}
Finally, we consider a series of tests and simulations to compare the performance of our approach with other approaches. We briefly discuss them below and refer readers to the Web Appendix for details.
\squishlist
\item In the simulations presented in the main text, we fix the hyper-parameters to default values. In Web Appendix $\S$\ref{app:hyp_expt}, we present additional results on hyper-parameter optimization and show how the choice of hyper-parameters can affect the performance of the approach. In particular, we show extensive results on how the relative information in decisions vs. state transitions affects the choice of $\lambda_{rel}$.

\item In Web Appendix $\S$\ref{subsec:RePaD_vs_CCP}, we present some comparisons against other approaches that can be adapted to high dimensional settings. First, we consider two-step estimators, where the first stage CCPs and state transitions are estimated non-parametrically or semi-parameterically \citep{Hotz_etal_1994, Norets_2012}. As discussed in $\S$\ref{subsec:curse_of_dimensionality}, defining and estimating state transitions is often infeasible in extremely high-dimensional settings. However, it is feasible to estimate CCPs flexibly and compare the accuracy of the CCPs obtained from using the raw set of high-dimensional state variables (i.e., using $x_{it}$, $q_{it}$) with those obtained from after partitioning based on our approach (i.e., using $x_{it}$, $\Pii(q_{it})$). We consider both flexible logits and neural networks for these first-stage estimators and show that across both types of estimators, the CCP estimates based on our partitioning approach consistently outperform (i.e., are more accurate than) those obtained by using the full set of state variables.

\item Finally, we consider two naive approaches to reduce the dimensionality of the high-dimensional state variables ($q_{it}$) -- $k$-means clustering and PCA. A key drawback of these approaches is that they reduce the dimensionality to capture the maximum variation in $q_{it}$. However, not all the variations in these variables may be relevant for the dynamic programming problem. That is, they are not able to address the fundamental challenge in our setting, which is to reduce the dimensionality of the state space while preserving decision-relevant information. Nevertheless, in Web Appendix $\S$\ref{subsec:RePaD_DimReduction}, we examine the ability of these two approaches to recover the true partitions in the data and show that they are unable to do so.
\squishend

\section{Applications, Extensions, Limitations, and Conclusion}
\label{sec:conclusion}
Finally, we present a discussion of other marketing settings where our approach can be applied, some potential extensions, and conclude with a brief summary of limitations.

\subsection{Applications}
\label{ssec:applications}
In the paper, we presented two sets of simulated experiments -- engine replacement problem and durable goods adoption -- to demonstrate the performance and applicability of our approach to high-dimensional dynamic discrete choice models. Beyond the technology and durable goods adoption problems discussed in $\S$\ref{ssec:finite_horizon}, there are many other natural settings where our approach is applicable. We now discuss a broader set of modern marketing applications that are likely to be suitable candidates for our method:
\squishlist
\item Promotion planning and pricing: In many online retail settings (e.g., Amazon, Instacart, Walmart), firms are interested in optimizing promotion/coupon delivery to customers. A large stream of marketing and economics literature has shown that customers tend to behave in a forward-looking fashion in these settings, i.e., they may choose to stockpile storable goods and wait for coupons/promotions to purchase \citep{Erdem_etal_2003, Hendel_Nevo_2006, VanHeerde_Neslin_2017}. In modern digital retail settings, firms have extensive information on product characteristics, consumers' demographics, and behavioral data on consumers' browsing and purchase behaviors, e.g., how often the consumer shops in the category, the set of prices/products they were exposed to during earlier shopping sessions, what products they subscribe to and the frequency of those subscriptions, what they searched for, and so on. Ideally, firms would like to leverage this information to optimize and personalize the cadence and depth of promotions given to consumers. However, it is unclear which (if any) of these different customer features affect utility and decision-making and how. Thus, in these settings, firms need DDC estimators that can incorporate high-dimensional consumer-level heterogeneity, and learn a lower-dimensional representation of customer features that best captures consumers' flow utilities and decisions. 

\item Software/services subscription and renewal: Many modern software and entertainment firms (e.g., Netflix, Spotify, DropBox, Adobe) rely on a Software as a Service (SaaS) model, where consumers pay a fixed monthly/yearly fee for using the service instead of pay-per-use pricing. In these settings, consumers often make forward-looking decisions, i.e., make the decision to commit or pay a fixed price now for a continued stream of future consumption. Firms have extensive information on consumers' prior and current product usage which can serve as a set of high-dimensional state variables that capture consumer-level heterogeneity, e.g., account demographics, frequency of consumption, type of content used/consumed, how many users within a household use the product, variety of content consumed, and so on. The firm may seek to optimize what products and features to offer to consumers and how to price different tiers of subscriptions and offerings. Again, in these settings, an in-depth understanding of how different consumer attributes and product features affect consumers' flow utility and inter-temporal decision-making can help firms optimize their marketing decisions. 


\item Mobile health apps: Mobile health apps and portals have taken off in the last few years and are now extensively used by lay consumers, patients, and healthcare providers (FitBit, MyChart, etc.). These apps and health-care providers are often interested in optimizing the delivery of reminders and recommendations for regular screenings, health appointments, and healthy behaviors \citep{mdotcenter, Liao_etal_2020}. Prior research has shown that consumers tend to trade off the immediate cost of engaging in healthy behaviors or getting medical screenings with the long-term benefits from these behaviors/screening \citep{Fang_Wang_2015}. Given the detailed data available on consumers' behavioral and health patterns, it is possible for these apps to leverage the full scale of these high-dimensional features to deliver personalized and appropriate communications/interventions (e.g., in-app notifications, reminders) to promote healthy behaviors. However, to do so, they will first need to estimate a dynamic discrete choice model of consumers' utility functions as a function of both high-dimensional consumer features/behaviors and physician/app-level communications and interventions. Our approach can help with specifying and estimating DDC models in this high-dimensional setting.


\squishend
Beyond the above applications, there are numerous other settings where modern firms have access to a large amount of features on consumer behavior and product attributes, as well as the ability to deliver real-time personalized marketing interventions. In many of these settings, consumers often make inter-temporal trade-offs, and there is significant heterogeneity in consumers' utilities as a function of their demographic and behavioral features. However, in many of these cases, we do not have extensive theory on which consumer attributes affect their utility and how, and the size of these features is often too large to rely on trial-and-error methods. Our estimation approach can help in these situations -- it can take a high-dimensional set of state variables as input and identify a lower-dimensional state-space that is relevant for consumer-decision-making, which can then be used to model consumers' utility functions and perform counterfactuals. 


\subsection{Extensions}
We now discuss two natural ways to improve the performance of the method and extend it to other cases.

\textbf{Cross-fitting:} The idea of cross-fitting has been used in Double Machine Learning approaches to avoid overfitting bias \citep{Chernozhukov_etal_2018, Chernozhukov_etal_2022}. Essentially, cross-fitting requires the researcher to fit the first-stage models on one dataset and then perform the estimation on a different dataset. In our setting, this can be done efficiently by splitting the data into two parts and using split-1 to learn the partitioning and performing the estimation on split-2 (taking the partition from split-1 as given) and vice versa. This separates the process of partitioning and estimation further, thereby reducing potential bias from overfitting/regularization during the first-stage partitioning. 


\textbf{Dynamic Games:} The method can be easily extended to dynamic games, e.g., entry models, oligopolistic models of competition and investment \citep{Bajari_etal_2007, Aguirregabiria_etal_2021}. In this case, in the recursive-partitioning algorithm, the state transition probabilities of an entity/agent would be conditioned on both their own actions as well as the actions of the other agents in the system. Then, once we have the partitioning, we can use a standard two-step estimator with the estimated partitions as another state variable, as we did in the case of the single-agent model (with the appropriate assumptions on the uniqueness of the equilibrium observed in the data). 

\subsection{Limitations and Conclusion}
Finally, we note that there are a series of limitations with our method, many of which are shared by other tree-based approaches. First, if $\Q$ enters the state transitions and decisions in a continuous fashion (instead of as partitions/categories), then no empirical partitioning will ever be perfect. This would be similar to using CART to capture a regression model where the explanatory variables enter regression in a continuous fashion, i.e., we can keep partitioning the data further and further without reaching a perfect discretization. In practice, this is not a big issue since the regularization procedure will stop the partitioning when the gains from partitioning are not sufficiently high to continue. However, in theory, there is likely to be some bias from the fact that we have not correctly accounted for the true partitions and have under-discretized $\Q$. Second, even though CART and other recursive partitioning methods are extremely popular for prediction and classification tasks and have many good properties, they lack some standard consistency properties. In particular, it is well-known that some data patterns cannot be captured by recursive partitioning even if the number of observations goes to infinity; see Web Appendix$\S$\ref{app:proof_con} and \citet{Biau_etal_2008} for details. Thus, there are no theoretical guarantees that this procedure will reach the true underlying partition for certain types of data-generating processes. Finally, from a practical perspective, we might not have enough observations to fully recover the accurate discretization when the dimensionality of $\Q$ is high, especially when we have a limited number of observations.



In practice, we can take a few steps to ensure that the partitioning algorithm is producing stable and reasonable results and is not unduly affected by these limitations. First, it is important to use hyper-parameter optimization to ensure that the discretization generated by our algorithm is generalizable to the data-generating process, not just the training sample. We can check the performance of estimated parameters in both training and validation sets and use the extent of the difference as a measure of the quality of our partitioning. Another suggestion is to check the sensitivity of the discretization and estimated parameters to the number of observations. For example, we can discretize $\Q$ and estimate the model with $80\%$ of the data, and check how close the discretization and estimated parameters are to the case where we use all the data. This could be a simple heuristic to test the robustness of the results.


In sum, our proposed algorithm allows researchers to reduce the estimation bias associated with ignoring relevant state variables, improve the fit of their models, and draw post-hoc inferences in a high-dimensional setting without theoretical guidance, all at relatively low compute costs. Further, because the recursive partitioning approach is modular, it does not require the user to develop or learn new estimation procedures. Indeed, it can be easily combined with any estimation procedure (e.g., nested fixed point, two-step) as a first-step approach. As such, we expect it to be easily applicable to a wide variety of settings.

\section*{Competing Interests Declaration}
Author(s) have no competing interests to declare.

\bibliographystyle{abbrvnat}
\bibliography{ref}

\newpage
\begin{appendices}

\setcounter{table}{0}
\setcounter{figure}{0}
\setcounter{equation}{0}
\setcounter{page}{0}
\renewcommand{\thetable}{A\arabic{table}}
\renewcommand{\thefigure}{A\arabic{figure}}
\renewcommand{\theequation}{A\arabic{equation}}
\renewcommand{\thepage}{\roman{page}}
\pagenumbering{roman}

\section{Proof of Convergence}
\label{app:proof_con}
In this section, we provide proof that our proposed algorithm converges to a perfect discretization. We start with a brief explanation of the structure of the proof.

The convergence proof consists of two steps. First, in \S\ref{subsec:likelihood_increase} we prove that the likelihood of observing data given partitioning $\Pii_r$ (at iteration $r$) and the optimal parameter value $\theta^*_r = \argmax_{\theta} \LL(\theta, \Pii_r)$ is increasing at each iteration of the algorithm. Therefore, each additional iteration does not reduce the likelihood of observing the data, i.e., $\LL(\theta^*_r, \Pii_r) \leq \LL(\theta^*_{r+1}, \Pii_{r+1})$. Then in \S\ref{subsec:likelihood_converge}, we prove that $\LL(\theta^*_r, \Pii_r) = \LL(\theta^*_{r+1}, \Pii_{r+1})$ if and only if $\Pii_r$ is a perfect discretization, proving that the proposed algorithm stops if and only if it reaches a perfect discretization. Before proceeding, we provide some definitions, assumptions, and lemmas that are used in the proof. 

\begin{definition}
We split the likelihood function into two parts: the decision part and the state transition part, denoted by $\LL_{dc}(\theta,\Pii)$ and $\LL_{st}(\Pii)$ respectively, such that $\LL(\theta,\Pii) = \LL_{dc}(\theta,\Pii) + \lambda \LL_{st}(\Pii)$. Note that the decision part of the likelihood is a function of the utility function parameters, $\theta$.
\end{definition}

\begin{definition}
Discretization $\Pii$ is a \textbf{parent} of discretization $\Pii'$ if for every $\pi' \in \Pii'$, there is a partition $\pi \in \Pii$ such that $\pi'$ is completely within $\pi$. In other words, discretization $\Pii$ is a \textit{parent} of discretization $\Pii'$ if one can generate discretization $\Pii'$ by further splitting discretization $\Pii$. 
\end{definition}

\begin{definition}
The expected value function and expected choice-specific value function are defined as follows:
\begin{align}
    \label{eq:value-function}
    \bar{V}(x,\pi;\theta, \Pii) &= \log \sum_{j \in \J} \exp{ v(x,\pi,j;\theta, \Pii) }
    \\
    \label{eq:choice-specific}
    v(x,\pi,j;\theta, \Pii) = \bar{u}(x,\pi,j;\theta, \Pii) &+ \beta \sum_{x' \in \X} \sum_{\pi' \in \Pii} \bar{V}(x',\pi';\theta, \Pii) g(x',\pi'|x,\pi,j;\theta,\Pii).
\end{align}
where $\beta$ is the discount factor usually set by the researcher. Assuming that error terms are drawn from Type 1 Extreme Value distribution, the predicted choice probabilities can be calculated as:
\begin{equation}
    \hat{p}(j|x,\pi;\theta, \Pii) = \frac{\exp{ v(x,\pi,j;\theta, \Pii)}}{\exp{ \bar{V}(x,\pi;\theta, \Pii)}}.
\end{equation}
\end{definition}

\begin{assumption}
\label{assumption:non_parametric_form}
We assume a fully non-parametric form for the flow utility and state transition function to avoid dependence on parametric functions for the proof and keep the analysis general, i.e., $\bar{u}(x,\pi,d;\theta, \Pii)$ and $g(x,\pi|x',\pi',j;\Pii)$ are constant for each combination of states, $\{x,\pi\}$, and decision $d$. 
\end{assumption}
\begin{assumption}
\label{assumption:no_bad_pattern}
There is no data pattern that is not discoverable with a single split.
\end{assumption}
Assumption \ref{assumption:no_bad_pattern} is common to all algorithms that are based on recursive partitioning. It has been shown that some patterns cannot be captured by recursive partitioning, even when the number of observations goes to infinity \citep{Biau_etal_2008}. A simple example of such patterns is presented in Figure \ref{figure:bad_pattern}. Assume observations are uniformly distributed throughout the covariate space. Also, assume that the flow utility in the crosshatched regions is $l$ and it is $h$ in the non-crosshatched regions. Any split would result in two sub-partitions with a similar number of observations in the crosshatched and non-crosshatched regions. The average utility would be equal to $\frac{h+l}{2}$ in the resulting two sub-partitions. Since the split does not increase $\LL_{dc}$, the algorithm does not add it to its discretization. Generally, recursive partitioning cannot capture variational patterns that are symmetric in a way that any splits would lead to two sub-partition with a similar average statistic. 

To overcome this shortcoming of recursive partitioning algorithms, we assume no data pattern exists that a single split in a discretization cannot partially capture. Then because of the assumption we can draw the following corollary.
\begin{corollary}
\label{col:assumption}
A discretization $\Pii$ is either perfect, or there exists a split such that the decision probabilities, average incoming transition probabilities, or outgoing transition probabilities on two sub-partitions are not equal. Formally, if discretization $\Pii$ is not perfect, there exists a split that partitions $\pi_p \in \Pii$ into $\pi_l$ and $\pi_r$ such that for a $\{x, x'\} \in \X$, $\pi' \in \Pii$ and $j \in \J$ at least one of the following inequalities holds: 
\begin{align*}
    \Pr(j|x,\pi_l) &\neq \Pr(j|x,\pi_r)\\
    \Pr(x',\pi'|x,\pi_l,j) &\neq\Pr(x',\pi'|x,\pi_r,j)\\
    \frac{\Pr(x,\pi_l|x',\pi',j)}{N(x,\pi_l)} &\neq \frac{\Pr(x,\pi_r|x',\pi',j)}{N(x,\pi_r)}\\
\end{align*}
\end{corollary}

\begin{figure}[!ht]
    \centering
    \includegraphics[width=70mm]{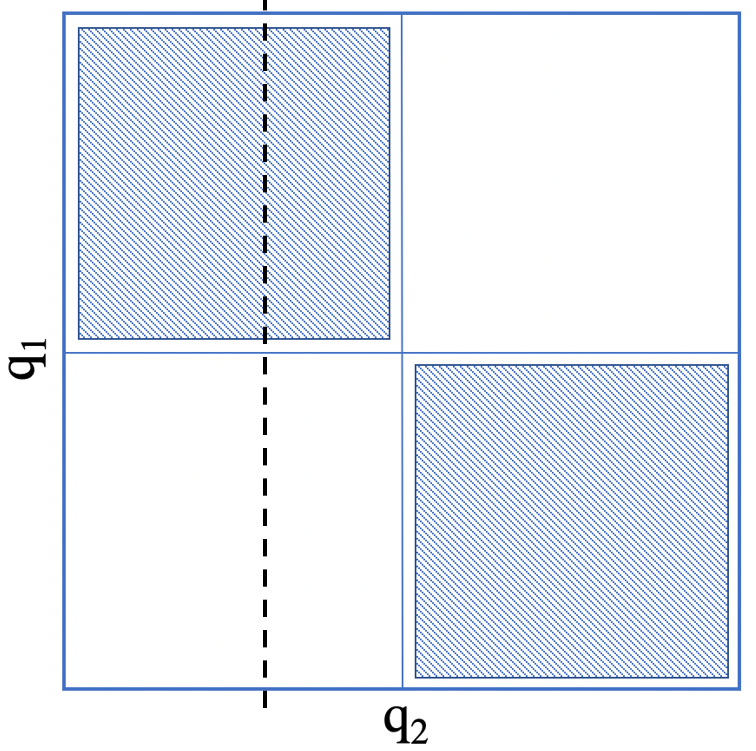}
    \caption{\label{figure:bad_pattern} An example of a pattern that cannot be captured by recursive partitioning. Observations in the white and crosshatched region have different statistics. Any split, such as the black dashed line, result in two sub-partitions that have similar average statistics.} 
\end{figure}

Finally, we need the following lemma for proving that the decision likelihood is increasing in each iteration in \S\ref{subsec:likelihood_increase}.
\begin{lemma}
\label{lem:bestprob}
For any given vector $a: \sum_i a_i = 1$, the solution to the following maximization problem is equal to $a$.
\begin{align*}
    \max_{b}& \quad f(b) = \sum_i a_i\ln{b_i} \\
    \textrm{s.t.}&\quad \sum_i b_i = 1
\end{align*}
\end{lemma}
\begin{proof}
It is a constrained optimization problem that can be solved by maximizing the Lagrangian function.
\begin{align*}
    \Lagr(a,b,\lambda) & =  \sum_i a_i\ln{b_i} - \lambda (\sum_i b_i - 1) \\
    \nabla \Lagr(b,\lambda) & =  0 \\
    \frac{\partial \Lagr}{\partial b_i} & = \frac{a_i}{b_i} - \lambda = 0 \Rightarrow a_i = \lambda b_i \Rightarrow \sum a_i = \lambda \sum b_i\\ 
    \frac{\partial \Lagr}{\partial \lambda} & = \sum_i b_i - 1 = 0 \Rightarrow \sum_i b_i = 1 \\
     &  \Rightarrow \lambda = 1 \Rightarrow b_i = a_i
\end{align*}
\end{proof}

\subsection{Proof of Likelihood Increase}
\label{subsec:likelihood_increase}
In this section, we prove that the likelihood function $\LL(\theta^*_r, \Pii_r)$ increases at each iteration of our recursive partitioning algorithm. Formally, we prove
\begin{equation}
        \LL(\theta^*_r, \Pii_r) \leq \LL(\theta^*_{r+1}, \Pii_{r+1})
\end{equation}
where $\Pii_r$ and $\Pii_{r+1}$ are the discretizations generated by the proposed recursive partitioning in iterations $r$ and $r+1$ respectively, $\theta^*_r = \argmax_{\theta} \LL(\theta, \Pii_r)$, and $\theta^*_{r+1} = \argmax_{\theta} \LL(\theta, \Pii_{r+1})$.

\begin{theorem}
\label{theorem:increasing}
For any candidate additional split, noted by $\{k,q,z\}$, to a discretization $\Pii$, where $k$ is a partition in $\Pii$, $q$ is a feature in $Q$, and $z$ is a value within the range of possible values for $q$ in $k$, the following inequalities hold
\begin{eqnarray}
\label{eq:decision_increase}
\exists \;\; \theta' : \LL_{dc}(\theta^*, \Pii) &\leq \LL_{dc}(\theta', \Pii') \\
\LL_{st}(\Pii) &\leq \LL_{st}(\Pii') \label{eq:transition_increase} 
\end{eqnarray}
where $\Pii' = \Pii + \{k,q,z\}$ is the discretization after adding the candidate split, and $\theta^* = \argmax_{\theta} \LL(\theta,\Pii)$ is the optimal parameters given discretization $\Pii$.
\end{theorem}

We prove Inequalities \ref{eq:decision_increase} and \ref{eq:transition_increase} separately. First, in Lemma \ref{lem:exp_form}, we create a new parameter set $\theta'$ for the discretization $\Pii'$ from $\theta^*$ that satisfies a certain inequality. Next, we show that Inequality \ref{eq:decision_increase} holds for the new parameter set $\theta'$ and $\Pii'$. Finally, we prove that for any two partitions $\Pii$ and $\Pii'$ such that $\Pii$ is the parent of $\Pii'$ the inequality \ref{eq:transition_increase} holds. Please note that discretization $\Pii$ is the parent of any discretization that is created by adding additional splits to $\Pii$.

\begin{lemma}
\label{lem:exp_form}
There is a flow utility coefficient set $\theta'$ such that for any $x \in \X$, $q \in \Q$ and $j \in \J$, the following equation holds
\begin{align}
    \notag
    \exp{v(x,\Pii'(q),j;\theta',\Pii')} &= \exp{v(x,\Pii(q),j;\theta^*,\Pii)} 
    \\ &+  \Big[\Pr(j|x,\Pii'(q)) - \Pr(j|x,\Pii(q))\Big]\exp{ \bar{V}(x,\Pii(q);\theta^*,\Pii)} 
    \label{eq:lemma_exp_form}
\end{align}
\end{lemma}
\begin{proof}
Using the equality $\log(a+b) = \log(a)+\log(1+b/a)$, we write Equation \eqref{eq:lemma_exp_form} as follows:
\begin{align}
\notag
    v(x,\Pii'(q),j;\theta',\Pii') &= v(x,\Pii(q),j;\theta^*,\Pii) + \log \Big ( 1 + \frac{\Pr(j|x,\Pii'(q)) - \Pr(j|x,\Pii(q))}{\frac{\exp{v(x,\Pi(q),j;\theta^*,\Pii)}}{\exp{\bar{V}(x,\Pi(q);\theta^*,\Pii) }}} \Big ) \\
    \label{eq:choice-specific-diff}
    & = v(x,\Pii(q),j;\theta^*,\Pii) + \log \Big ( 1 + \frac{\Pr(j|x,\Pii'(q)) - \Pr(j|x,\Pii(q))}{\hat{p}(j|x,\Pii(q);\theta^*,\Pii)} \Big )
\end{align}

Using the formulation of the choice-specific value function (Equation \eqref{eq:choice-specific}), we can write Equation \eqref{eq:choice-specific-diff} in terms of the flow utility and expected value functions as follows:
\begin{align}
\notag
    u(x,\Pii'(q),j;\theta',\Pii') &= u(x,\Pii(q),j;\theta^*,\Pii) \\ 
    \notag
    & + \beta \sum_{x' \in X} \sum_{\pi \in \Pii} \bar{V}(x',\pi;\theta^*,\Pii) g(x',\pi|x,\Pii(q),j)  \\
    \notag
    & - \beta \sum_{x' \in X} \sum_{\pi \in \Pii'} \bar{V}(x',\pi;\theta',\Pii')  g(x',\pi|x,\Pii'(q),j) \\
    \label{eq:new-utility-equation-tmp}
    & + \log \Big ( 1 + \frac{\Pr(j|x,\Pii'(q)) - \Pr(j|x,\Pii(q))}{\hat{p}(j|x,\Pii(q);\theta^*,\Pii)} \Big )
\end{align}
We assume a non-parametric functional form for the utility function; therefore, we can calculate a set of new utility values for each observable state $\{x,\pi\}$ based on the above equation that satisfies Equation \ref{eq:lemma_exp_form}. However, the right-hand side of Equation \eqref{eq:new-utility-equation-tmp} uses $\theta'$, which we are trying to calculate. If we want to use Equation \eqref{eq:new-utility-equation-tmp}, we need to prove that this equation has a unique answer. Alternatively, we prove that if Equation \ref{eq:lemma_exp_form} holds, the value functions in the new partitioning and old partitioning are equal.
\begin{align}
\notag
    \bar{V}&(x,\Pii'(q);\theta',\Pii') \\
    \notag
    & = \log \sum_{j \in \J} \exp{v(x,\Pii'(q),j;\theta',\Pii')} \\
    \notag
     &= \log \sum_{j \in \J} \Big ( \exp{v(x,\Pii(q),j;\theta^*,\Pii)} +  \big[\Pr(j|x,\Pii'(q)) - \Pr(j|x,\Pii(q))\big]\bar{V}(x,\Pii(q);\theta^*,\Pii)\Big )\\
     \notag
     & = \log \Big ( \sum_{j \in \J}  \exp{v(x,\Pii(q),j;\theta^*,\Pii)} + \bar{V}(x,\Pii(q);\theta^*,\Pii) \sum_{j \in \J}  \big[\Pr(j|x,\Pii'(q)) - \Pr(j|x,\Pii(q))\big]\Big )\\
     \notag
     & = \log \sum_{j \in \J}  \exp{v(x,\Pii(q),j;\theta^*,\Pii)} \\
     \label{eq:value-function-equal}
     & = \bar{V}(x,\Pii(q);\theta^*,\Pii) 
\end{align}
where the equality from line $3^{rd}$ to line $4^{th}$ comes from $\sum_{j \in \J}  \big[\Pr(j|x,\Pii'(q)) - \Pr(j|x,\Pii(q))\big] = 0$.

We can now write Equation \eqref{eq:new-utility-equation-tmp} as:
\begin{align}
\notag
    u(x,\Pii'(q),j;\theta',\Pii') &= u(x,\Pii(q),j;\theta^*,\Pii) \\ 
    \notag
    & + \beta \sum_{x' \in X} \sum_{\pi \in \Pii} \bar{V}(x',\pi;\theta^*,\Pii) g(x',\pi|x,\Pii(q),j)  \\
    \notag
    & - \beta \sum_{x' \in X} \sum_{\pi \in \Pii'} \bar{V}(x',\Pii(\pi);\theta^*,\Pii)  g(x',\pi|x,\Pii'(q),j) \\
    \label{eq:new-utility-equation}
    & + \log \Big ( 1 + \frac{\Pr(j|x,\Pii'(q)) - \Pr(j|x,\Pii(q))}{\hat{p}(j|x,\Pii(q);\theta^*,\Pii)} \Big )
\end{align}
where $\Pii(\pi)$ is the partition in $\Pii$ that includes all $\pi \in \Pii'$. Equation \ref{eq:new-utility-equation} is not dependent on $\theta'$, and can be used to calculate a set of utility function values ($\theta'$) for partitioning $\Pii'$. We can show that if one substitute $u(x,\Pii'(q),j;\theta',\Pii')$ from Equation \ref{eq:new-utility-equation} into the left hand side of Equation \ref{eq:lemma_exp_form}, then we should see that the Equation holds. Essentially, we are deriving a $\theta'$ from the current partition $\Pii$ and $\theta^*$ such that Equation \ref{eq:lemma_exp_form} holds.
\end{proof}

\begin{proof}[Proof of Inequality \eqref{eq:decision_increase}]
We show in Lemma \ref{lem:exp_form} that there is a $\theta'$ such that Equality \ref{eq:lemma_exp_form} holds, and for any $x \in X$ and $q \in \Q$ we have $\bar{V}(x,\Pii'(q);\theta',\Pii')=\bar{V}(x,\Pii(q);\theta^*,\Pii)$. 
Since $\bar{V}(.;\theta^*,\Pii)$ is a fixed point solution to the standard contraction mapping problem for the Bellman equation, $\bar{V}(.;\theta',\Pii')$ generated in Lemma \ref{lem:exp_form} is a solution to the contraction mapping in Bellman equation as well. Additionally, we have

\begin{align}\notag
    \LL_{dc}(\theta',\Pii') &= \sum_{i=1}^N \sum_{t=1}^T \log \hat{p}(d_{it}|x_{it},\Pii'(q_{it});\theta',\Pii') \\
\notag
    & = \sum_{i=1}^N \sum_{t=1}^T\log \frac{\exp{v(x_{it},\Pii(q_{it}),d_{it};\theta',\Pii')}}{\exp{\bar{V}(x_{it},\Pii(q_{it});\theta',\Pii') }} \\
\notag
    & = \sum_{i=1}^N \sum_{t=1}^T\log \frac{\exp{v(x_{it},\Pii(q_{it}),d_{it};\theta',\Pii')}}{\exp{\bar{V}(x_{it},\Pii(q_{it});\theta^*,\Pii) }} \\
    \label{eq:llproofstep}
    & = \sum_{i=1}^N \sum_{t=1}^T\log \Big ( \frac{\exp{v(x_{it},\Pii(q_{it}),d_{it};\theta^*,\Pii)}}{\exp{\bar{V}(x_{it},\Pii(q_{it});\theta^*,\Pii) }} + \Pr(d_{it}|x_{it},\Pii'(q_{it})) - \Pr(d_{it}|x_{it},\Pii(q_{it})) \Big )
\end{align}

Please note the equality between line two and line three is derived from Equation \eqref{eq:value-function-equal}. Also, we use Equation \eqref{eq:lemma_exp_form} to go from three and line four. Using $\log(a+b) = \log(a) + \log(1+b/a)$, we have the following
\begin{align}
\notag
    & \LL_{dc}(\theta',\Pii')  = \sum_{i=1}^N \sum_{t=1}^T\log \Big ( \frac{\exp{v(x_{it},\Pii(q_{it}),d_{it};\theta^*,\Pii)}}{\exp{\bar{V}(x_{it},\Pii(q_{it});\theta^*,\Pii) }} + \Pr(d_{it}|x_{it},\Pii'(q_{it})) - \Pr(d_{it}|x_{it},\Pii(q_{it})) \Big ) \\
    \notag
    & = \sum_{i=1}^N \sum_{t=1}^T \log \Big ( \frac{\exp{v(x_{it},\Pii(q_{it}),d_{it};\theta^*,\Pii)}}{\exp{\bar{V}(x_{it},\Pii(q_{it});\theta^*,\Pii) }}  \Big ) + \sum_{i=1}^N \sum_{t=1}^T \log \Big ( 1 + \frac{\Pr(d_{it}|x_{it},\Pii'(q_{it})) - \Pr(d_{it}|x_{it},\Pii(q_{it}))}{\frac{\exp{v(x_{it},\Pii(q_{it}),d_{it};\theta^*,\Pii)}}{\exp{\bar{V}(x_{it},\Pii(q_{it});\theta^*,\Pii) }}} \Big ) \\
    \notag
    & =  \LL_{dc}(\theta^*,\Pii) + \sum_{i=1}^N \sum_{t=1}^T \log \Big ( 1 + \frac{\Pr(d_{it}|x_{it},\Pii'(q_{it})) - \Pr(d_{it}|x_{it},\Pii(q_{it}))}{\hat{p}(d_{it}|x_{it},\Pii(q_{it});\theta^*,\Pii)} \Big ) \\
    &    \label{eq:likelihooddiff} \Rightarrow
    \LL_{dc}(\theta',\Pii') - \LL_{dc}(\theta^*,\Pii) = \sum_{i=1}^N \sum_{t=1}^T \log \Big ( 1 + \frac{\Pr(d_{it}|x_{it},\Pii'(q_{it})) - \Pr(d_{it}|x_{it},\Pii(q_{it}))}{\hat{p}(d_{it}|x_{it},\Pii(q_{it});\theta^*,\Pii)} \Big )
\end{align}
We will show that the right-hand side of Equation \ref{eq:likelihooddiff} is positive; thus, proving that the likelihood is increasing. Assuming $NT \to \infty$, and our predicted choice probabilities are consistent, concludes $\hat{p}(d_{it}|x_{it},\Pii(q_{it});\theta^*,\Pii)=\Pr(d_{it}|x_{it},\Pii(q_{it}))$\footnote{We can also conclude this equality since we assume a fully non-parametric form for utility function in Assumption \ref{assumption:non_parametric_form}.}. Replacing the predicted choice probability with its counter-part conditional choice probability in Equation \ref{eq:likelihooddiff} yields
\begin{align}
\notag
      &\sum_{i=1}^N \sum_{t=1}^T \log \Big ( 1  + \frac{\Pr(d_{it}|x_{it},\Pii'(q_{it})) - \Pr(d_{it}|x_{it},\Pii(q_{it}))}{\hat{p}(d_{it}|x_{it},\Pii(q_{it});\theta^*,\Pii)} \Big ) \\
\notag
      & = \sum_{i=1}^N \sum_{t=1}^T \log \Big ( 1 + \frac{\Pr(d_{it}|x_{it},\Pii'(q_{it})) - \Pr(d_{it}|x_{it},\Pii(q_{it}))}{\Pr(d_{it}|x_{it},\Pii(q_{it}))} \Big) 
    \\ 
    \label{eq:probdiv}
    & = \sum_{i=1}^N \sum_{t=1}^T \log \frac{\Pr(d_{it}|x_{it},\Pii'(q_{it}))}{\Pr(d_{it}|x_{it},\Pii(q_{it}))}
\end{align}

The value inside the summation in Equation \eqref{eq:probdiv} is equal to zero for all the observations, except for those where $q_{it}$ lands in the newly split partition. Let us call the partition before the split $\pi_p \in \Pii$, and the two resulting partitions $\{\pi_l, \pi_r\} \in \Pii'$. Also let $N(x,\pi)$ denote the number of observations where $x_{it} = x$ and $\Pii'(q_{it}) = \pi$, and $N(x,\pi,j)$ denote the number where in addition to the aforementioned conditions $d_{it} = j$. We can write Equation \eqref{eq:probdiv} as follows:
\begin{align*}
    \sum_{i=1}^N \sum_{t=1}^T &\log \frac{\Pr(d_{it}|x_{it},\Pii'(q_{it}))}{\Pr(d_{it}|x_{it},\Pii(q_{it}))}  = \sum_{x \in X} \sum_{\pi \in \{\pi_l,\pi_r\}} \sum_{j \in \J} N(x,\pi,j) \log \big ( \Pr(j|x,\pi) - \Pr(j|x,\pi_p) \big)
    \\ & =  \sum_{x \in X} \sum_{\pi \in \{\pi_l,\pi_r\}} N(x,\pi) \sum_{j \in \J} \big ( \Pr(j|x,\pi) \log \Pr(j|x,\pi) - \Pr(j|x,\pi) \log \Pr(j|x,\pi_p) \big )
\end{align*}
According to Lemma \ref{lem:bestprob}, $\sum_{j \in \J} \Pr(j|x,\pi) \log \Pr(j|x,\pi) \geq \sum_{j \in \J} \Pr(j|x,\pi) \log \Pr(j|x,\pi_p)$. This inequality is strict if there is a $j \in \J$ such that $\Pr(j|x,\pi)\neq \Pr(j|x,\pi_p)$ for any $x \in X$ and $\pi \in \{\pi_l, \pi_r\}$. Thus, Equation \eqref{eq:probdiv} is greater than or equal to zero, which proves the Inequality \eqref{eq:decision_increase} in Theorem \ref{theorem:increasing}.
\end{proof}


\begin{proof}[Proof of Inequality \eqref{eq:transition_increase}] First, note that discretization $\Pii$ is a parent of discretization $\Pii'$. Based on the definition of parent, for every $\pi \in \Pii$, there is a set of partitions $\{\pi_i\} \in \Pii'$ such that $\bigcup\limits_{i} \pi_i = \pi$. Let us call $\{\pi_i\}$ the child set of $\pi$ in $\Pii'$ and denote it by $\Pii'(\pi)$. First we use the log sum inequality \citep{Cover_Thomas_1991} to prove that for any $\{\pi,\pi'\} \in \Pii$, $\{x,x'\} \in \X$ , and $j \in \J$ the following inequality holds:
\begin{align}
\label{ineq:logsumtotal}
\sum_{\pi_i \in \Pii'(\pi)} \sum_{\pi'_i \in \Pii'(\pi')}  N(x,\pi_i,x',\pi'_i,j) \log \frac{N(x,\pi_i,x',\pi'_i,j)}{N(x,\pi_i)N(x',\pi'_i,j)} \geq N(x,\pi,x',\pi',j) \log \frac{N(x,\pi,x',\pi',j)}{N(x,\pi)N(x',\pi',j)}
\end{align}

We prove this inequality by applying the log sum inequality twice. First for a given $\pi_i \in \Pii'(\pi)$ the following holds according to log sum inequality\footnote{We have $\sum_{\pi'_i \in \Pii'(\pi')}  N(x,\pi_i,x',\pi'_i,j) = N(x,\pi_i,x',\pi',j)$, and $\sum_{\pi'_i \in \Pii'(\pi')}  N(x',\pi'_i,j) = N(x',\pi',j)$, }.
\begin{align}
\label{eq:logsum_inequality_1}
\sum_{\pi'_i \in \Pii'(\pi')}  N(x,\pi_i,x',\pi'_i,j) \log \frac{N(x,\pi_i,x',\pi'_i,j)}{N(x',\pi'_i,j)} \geq N(x,\pi_i,x',\pi',j) \log \frac{N(x,\pi_i,x',\pi',j)}{N(x',\pi',j)}
\end{align}
which by subtracting $N(x,\pi_i,x',\pi',j) \log N(x,\pi_i)$ from both sides changes to
\begin{align}
\label{ineq:logsum1}
\sum_{\pi'_i \in \Pii'(\pi')}  N(x,\pi_i,x',\pi'_i,j) \log \frac{N(x,\pi_i,x',\pi'_i,j)}{N(x',\pi'_i,j) N(x,\pi_i)} \geq N(x,\pi_i,x',\pi',j) \log \frac{N(x,\pi_i,x',\pi',j)}{N(x',\pi',j)N(x,\pi_i)}.
\end{align}

Similarly, according to log sum inequality we have 
\begin{align}
\label{eq:logsum_inequality_2}
\sum_{\pi_i \in \Pii'(\pi)}  N(x,\pi_i,x',\pi',j) \log \frac{N(x,\pi_i,x',\pi',j)}{N(x,\pi_i)} \geq N(x,\pi,x',\pi',j) \log \frac{N(x,\pi,x',\pi',j)}{N(x,\pi)}
\end{align}
which by subtracting $N(x,\pi,x',\pi',j) \log N(x',\pi',j)$ from both sides changes to 
\begin{align}
\label{ineq:logsum2}
\sum_{\pi_i \in \Pii'(\pi)}  N(x,\pi_i,x',\pi',j) \log \frac{N(x,\pi_i,x',\pi',j)}{N(x,\pi_i)N(x',\pi',j)} \geq N(x,\pi,x',\pi',j) \log \frac{N(x,\pi,x',\pi',j)}{N(x,\pi)N(x',\pi',j)}.
\end{align}

By merging inequalities \ref{ineq:logsum1} and \ref{ineq:logsum2} we have
\begin{align*}
\sum_{\pi_i \in \Pii'(\pi)} \sum_{\pi'_i \in \Pii'(\pi')}  & N(x,\pi_i,x',\pi'_i,j) \log \frac{N(x,\pi_i,x',\pi'_i,j)}{N(x,\pi_i)N(x',\pi'_i,j)} \\ &\geq \sum_{\pi_i \in \Pii'(\pi)}  N(x,\pi_i,x',\pi',j) \log \frac{N(x,\pi_i,x',\pi',j)}{N(x,\pi_i)N(x',\pi',j)} \\ &\geq N(x,\pi,x',\pi',j) \log \frac{N(x,\pi,x',\pi',j)}{N(x,\pi)N(x',\pi',j)}
\end{align*}
which concludes inequality \ref{ineq:logsumtotal}. We can prove the lemma by summing this inequality over all $\{\pi,\pi'\} \in \Pii$, $\{x,x'\} \in \X$ and $j \in \J$.
\end{proof}
We proved both inequalities in the theorem \ref{theorem:increasing}. Next we prove that the equality happens if and only if our algorithm reaches a perfect discretization.
\subsection{The proof of convergence to perfect discretization}
\label{subsec:likelihood_converge}
In this section we prove that our proposed algorithm stops once it reaches a perfect discretization.
\begin{theorem}
\label{theorem:convergence}
The recursive partitioning algorithm discussed in $\S$\ref{sec:our_approach} stops once it find a perfect discretization, i.e., $\FF(\Pii_{r})=\FF(\Pii_{r+1})$ if and only if $\Pii_{r}$ is a perfect discretization.
\end{theorem}
We prove this theorem by proving separate lemmas for decision and state transition probabilities. Let us denote the decision and transition part of $\FF(\Pii)$ by $\FF_{dc}(\Pii)$ and $\FF_{tr}(\Pii)$ respectively. 
\begin{lemma}
\label{lem:ff_decision_increasing}
For any candidate additional split $\{k,q,z\}$ to discretization $\Pii$, that splits $\pi_p \in \Pii$ into $\{\pi_l,\pi_r\} \in \Pii'$, we have $\FF_{dc}(\Pii)=\FF_{dc}(\Pii')$ if and only if for all $x \in \X$ and $j \in \J$ the decision probabilities in $\pi_l$ and $\pi_r$ are similar, i.e., $\Pr(j|x,\pi_l) = \Pr(j|x,\pi_r)$.
\end{lemma}
\begin{proof}
We have
\begin{align*}
    \FF_{dc}(\Pii') - \FF_{dc}(\Pii) &= N(x,\pi_l,j; \Pii) \log \frac{N(x,\pi_l,j; \Pii)}{N(x,\pi_l; \Pii)} + N(x,\pi_r,j; \Pii) \log \frac{N(x,\pi_r,j; \Pii)}{N(x,\pi_r; \Pii)} 
    \\ &- N(x,\pi_p,j; \Pii) \log \frac{N(x,\pi_p,j; \Pii)}{N(x,\pi_p; \Pii)}
\end{align*}
According to log sum inequality the right-hand side of this equality is equal to zero if and only if $\frac{N(x,\pi_r,j; \Pii)}{N(x,\pi_r; \Pii)} = \frac{N(x,\pi_l,j; \Pii)}{N(x,\pi_l; \Pii)}$, which concludes $\Pr(j|x,\pi_l) = \Pr(j|x,\pi_r)$.
\end{proof}

\begin{lemma}
\label{lem:parent_ff_transition_increasing}
For any candidate additional split $\{k,q,z\}$ to discretization $\Pii$, that splits $\pi_p \in \Pii$ into $\{\pi_l,\pi_r\} \in \Pii'$, we have $\FF_{tr}(\Pii)=\FF_{tr}(\Pii')$ if and only if for all  $\{x,x'\} \in \X$, $\pi' \in \Pii'$, and $j \in \J$ the following equations hold
\begin{align*}
    \Pr(x',\pi'|x,\pi_l,j)&=\Pr(x',\pi'|x,\pi_r,j)\\
    \frac{\Pr(x,\pi_l|x',\pi',j)}{N(x,\pi_l)}&=\frac{\Pr(x,\pi_r|x',\pi',j)}{N(x,\pi_r)}
\end{align*}
\end{lemma}
\begin{proof}

\begin{figure}[!ht]
    \centering
    \subfloat[Transition from the new sub-partitions ]{{\includegraphics[width=7.5cm]{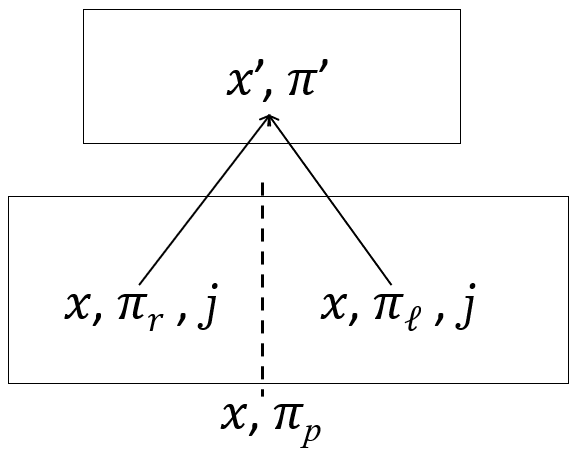} }}
    \qquad
    \subfloat[Transition to the new sub-partitions]{{\includegraphics[width=7.5cm]{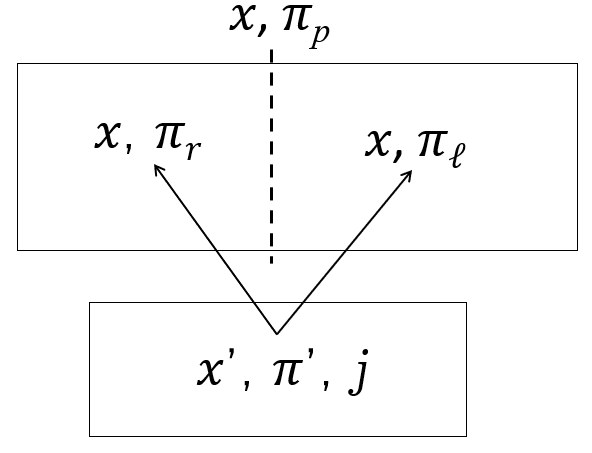} }}
    \caption{\label{figure:transition_proof} The change in likelihood from transition-to and transition-from perspective.} 
\end{figure}

We proved in Web Appendix\ref{subsec:likelihood_increase} that likelihood, $\LL(\theta_t)$, increases with any additional split. We assumed a completely non-parametric form for the state transition part of the likelihood function. Therefore the state transition of the likelihood function and recursive partitioning objective function are the same, i.e., $\LL_{tr}(\Pii)=\FF_{tr}(\Pii)$. Here we prove that the increment in likelihood, and consequently in $\FF(\theta_t)$, is equal to zero if and only if Lemma's equations hold. The split changes the likelihood by changing the transition-to and average transition-from probabilities presented in figure \ref{figure:transition_proof}. We can calculate the changes in likelihood with respect to each of these changes separately. First, according to (a) in figure \ref{figure:transition_proof} we have
\begin{align*}
    \FF(\Pii')-\FF(\Pii) = \sum_{x,x' \in \X} \sum_{\pi' \in \Pii'} \sum_{j \in \J} N(x',\pi') \Bigg ( & \frac{N(x',\pi',x,\pi_r,j)}{N(x',\pi')} \log \frac{N(x',\pi',x,\pi_r,j)}{N(x',\pi')N(x,\pi_r,j)}
    \\  + &\frac{N(x',\pi',x,\pi_l,j)}{N(x',\pi')} \log \frac{N(x',\pi',x,\pi_l,j)}{N(x',\pi')N(x,\pi_l,j)} 
    \\- &\frac{N(x',\pi',x,\pi_p,j)}{N(x',\pi')} \log \frac{N(x',\pi',x,\pi_p,j)}{N(x',\pi')N(x,\pi_p,j)} \Bigg )
\end{align*}

According to log sum inequality the term in the parentheses is greater or equal to zero. The left-hand side is equal to zero if and only if $\frac{N(x',\pi',x,\pi_l,j)}{N(x,\pi_l,j)} = \frac{N(x',\pi',x,\pi_r,j)}{N(x,\pi_r,j)} = \frac{N(x',\pi',x,\pi_p,j)}{N(x,\pi_p,j)}$, for every  $\{x,x'\} \in \X$, $j \in \J$ and $\pi' \in \Pii'$. This concludes the first equation of the lemma.

We can conclude the second equation similarly with (b) in figure \ref{figure:transition_proof} as the following
\begin{align*}
    \FF(\Pii')-\FF(\Pii) = \sum_{x,x' \in \X} \sum_{\pi' \in \Pii'} \sum_{j \in \J}  N(x',\pi', j) \Bigg ( & \frac{N(x,\pi_r,x',\pi',j)}{N(x',\pi',j)} \log \frac{N(x,\pi_r,x',\pi',j)}{N(x,\pi_r)N(x',\pi',j)}
    \\  + &\frac{N(x,\pi_l,x',\pi',j)}{N(x',\pi',j)} \log \frac{N(x,\pi_l,x',\pi',j)}{N(x,\pi_l)N(x',\pi',j)} 
    \\- &\frac{N(x,\pi_p,x',\pi',j)}{N(x',\pi',j)} \log \frac{N(x,\pi_p,x',\pi',j)}{N(x,\pi_p)N(x',\pi',j)}\Bigg )
\end{align*}

Again, according to log sum inequality the term in the parentheses is greater or equal to zero. The left-hand side is equal to zero if and only if $\frac{N(x',\pi',x,\pi_l,j)}{N(x,\pi_l)N(x',\pi',j)} = \frac{N(x',\pi',x,\pi_r,j)}{N(x,\pi_r)N(x',\pi',j)} = \frac{N(x',\pi',x,\pi_p,j)}{N(x,\pi_p)N(x',\pi',j)}$, for every  $\{x,x'\} \in \X$, $j \in \J$ and $\pi' \in \Pii'$. This concludes the second equation of the lemma.
\end{proof}
\begin{proof}[Proof for theorem \ref{theorem:convergence}]
Now we prove theorem \ref{theorem:convergence} using lemmas \ref{lem:parent_ff_transition_increasing} and \ref{lem:ff_decision_increasing}. Assume that our algorithm stops at iteration $r$. Given that the algorithm stops once no split increase the $\FF$ function, for any additional split that generates candidate partition $\Pii'$, we have $\FF(\Pii') = \FF(\Pii_r)$. Therefore, according to lemmas \ref{lem:parent_ff_transition_increasing} and \ref{lem:ff_decision_increasing} for any additional split to $\Pii_r$ the following equalities hold for all $\{x,x'\} \in \X$, $\pi' \in \Pii_r$ and $j \in \J$.
\begin{align*}
    \Pr(d|x,\pi_l) &= \Pr(d|x,\pi_r)\\
    \Pr(x',\pi'|x,\pi_l,j) &=\Pr(x',\pi'|x,\pi_r,j)\\
    \frac{\Pr(x,\pi_l|x',\pi',j)}{N(x,\pi_l)} &= \frac{\Pr(x,\pi_r|x',\pi',j)}{N(x,\pi_r)}\\
\end{align*}
which concludes that $\Pii_r$ is a perfect discretization according to corollary \ref{col:assumption}.
\end{proof}

\section{Hyper-parameter Optimization Experiments}
\label{app:hyp_expt}
\subsection{Experiment 1: Scoring function comparison}
We now conduct a series of numerical experiments to examine the importance of hyperparameter tuning. Specifically, we focus on the choice of $\lambda_{rel}$, which is unique to our dynamic setting and captures the relative importance of the choice probabilities and state transitions in designing the partition. Recall that when $\lambda_{rel}$ is large, the recursive partitioning algorithm puts a higher weight on state transitions (compared to choice probabilities) and vice versa. Tuning this hyperparameter should, therefore, allow us to weigh the state transition data less when state transitions are noisy and/or less informative compared to choice probabilities (since higher $\lambda_{rel}$ in such cases would lead the algorithm to pick up noise as a signal in its discretization). 

In general, we expect the choice of $\lambda_{rel}$ to be important when: (1) the dimensionality of $\Q$ is high since it is much easier for the algorithm to find noisy patterns in such cases (because there are a lot more variables for a potential split), and (2) when the number of observations is relatively small, i.e., finite samples, because there may not be enough data in each possible partition to pick up patterns accurately. In the rest of this section, we therefore explore the sensitivity of the partitioning procedure to the choice of $\lambda_{rel}$ for different data-generating processes and data sizes. 

\subsubsection{Data generating process} 
For these experiments, we use the Rust bus engine setting. Therefore, the data-generating process in these simulations is similar to the first study, with small differences, as described below.

\begin{enumerate}
\item First, the dimensionality of $\Q$ is 30, but only the first 10 variables affect the data-generating process, and the rest are irrelevant. We randomly discretize $\Q$ into 15 partitions using the process explained in Web Appendix$\S$\ref{ssec:random_partition_generator}. We use this process to generate 100 different discretizations, i.e., 100 cases, each with a different discretization of $\Q$. 

\item We then allow the replacement cost in each partition in each case to be a function of the partition in $\Pii^*$ as shown in Web Appendix$\S$\ref{ssec:random_partition_generator}. 
\end{enumerate}

Next, for each of the 100 discretizations, we vary the mileage transition model, the state transition model in $\Q$, and the amount of data available in two ways each, which gives us eight different data-generating scenarios for each of the discretizations. These are discussed below. 
\begin{enumerate}[resume]
\item We consider two cases for $f_{tr}$: (i) $f_{tr}(\pi) = 1$ for all 15 partitions (Similar mileage transition case), and (ii) $f_{tr}(\pi)$ is a random number from the set $\{0,1,2,3\}$ (Dissimilar mileage transition case). 
\item We consider two types of state transition models for $\Q$:
 \squishlist
    \item Random transition: Agents' transitions in the $\Q$-space are completely random. In each period, an agent randomly transitions from one partition to another such that: $\Pii^*(q_{it}) \perp \Pii^*(q_{it+1})$.

    \item Sparse transition: After each period, the agent remains in the same partition with probability $1/3$ and moves to one of the next two partitions with the same probability, and the ordinality/ordering of partitions is randomly chosen in each simulation.

\squishend
Variation across these two dimensions allows us to vary the amount of information available in the state transition. For example, there is less information in the state transition in the case where the transitions across partitions in $\Pii^*$ random and $f_{tr}(\pi) = 1$.

\item We also vary the amount of data available for the analysis by considering two scenarios for the number of observations: (i) 100 buses in 100 periods and (ii) 100 buses in 400 periods. 
\end{enumerate}

For all these scenarios, we examine the impact of different values of $\lambda_{rel}$ by learning the partition from the training data and evaluating it on a validation dataset (that is separately generated and is the same size as the training data) using the score function (as shown in Equation \eqref{eq:score}). Specifically, we consider the following values of $\lambda_{rel} \in \{0,0.2,0.5,1,2,5,100\}$. During this analysis, we fix all the other hyper-parameters as follows: {\it minimum number of observations} = 1, {\it minimum lift} = $10^{-10}$, and {\it total number of partitions} = 15 (which is the number of true partitions in the data). Fixing all the other hyper-parameters allows us to focus on the role of $\lambda_{rel}$. 


\begin{table}[t]
\centering
\footnotesize
\begin{tabular}{c|ccc|ccccccc}
\toprule
Case & Number & Transition & $\Pii^*(Q)$  & \multicolumn{7}{c}{\multirow{2}{*}{Value of $\lambda_{rel}$}}\\
No. & of & in $\Pii^*(Q)$ & affects mileage \\
& Periods  &  space & transition &    0 &    0.2 & 0.5 &  1 &  2 &    5 &  100 \\
\midrule
Case 1 & \multirow{4}{*}{100} & Sparse & Yes & -3778 & -3530 & -3328 & -3151 & -3016 & -2906 & \textbf{-2811} \\
Case 2 &  & Sparse & No & -3772 & -3559 & -3429 & -3343 & -3267 & \textbf{-3224} & -3233 \\
Case 3 &  & Random & Yes & \textbf{-3783} & -3803 & -3786 & -3821 & -3875 & -3972 & -4125 \\
Case 4 &  & Random & No & \textbf{-3715} & -3794 & -3880 & -3973 & -4078 & -4208 & -4389 \\
    \midrule
Case 5 & \multirow{4}{*}{400} & Sparse & Yes & -13478 & -13179 & -12796 & -12375 & -11994 & -11664 & \textbf{-11393} \\
Case 6 &  & Sparse & No & -13839 & -13696 & -13516 & -13323 & -13104 & \textbf{-12994} & -13042 \\
Case 7 &  & Random & Yes & \textbf{-13634} & -13754 & -13854 & -13928 & -13993 & -14158 & -14355 \\
Case 8 &  & Random & No & \textbf{-13673} & -13987 & -14303 & -14619 & -14936 & -15264 & -15643 \\
\bottomrule
\end{tabular}
\caption{\label{table:simulation_2_res_score} The calculated score for different values of $\lambda_{rel}$ in different data-generating processes. A bigger $\lambda_{rel}$ is better when there is more information in the state transition data. Scores are averaged over the 100 different partitioning schemes and shown for the validation data.}
\end{table}


\subsubsection{Results} 
In total, we run $100$ (partitions) $\times 8$ (scenarios)$\times 7$ (values for $\lambda_{rel}$) $= 5600$ simulations and report the score on the validation data in Table  \ref{table:simulation_2_res_score}. The bold numbers in each row denote the best average score (on the validation data) across all the 100 partitions for a given scenario. We find that the optimal value of $\lambda_{rel}$ varies across data-generating processes. As expected, the optimal value for $\lambda_{rel}$ is the largest when the state transition is most informative (i.e. cases 1 and 5, where the transition in $\Pii^*(\Q)$  is Sparse (as opposed to Random), and the partition affects the mileage transition). Similarly, the optimal value $\lambda_{rel}$ is small (or zero) in cases 4 and 8, where the transition in $\Pii^*(\Q)$ is Random and $\Pii^*(\Q)$ does not affect mileage transition (thus not informative for finding the discretization).

In summary, these simulations demonstrate that using the state transition information in the recursive partitioning algorithm helps improve the quality of the partitions, though the extent to which this data helps varies with the data-generating process (and depends on the level of signal available in the state transitions). Thus, it is important to tune $\lambda_{rel}$ in a data-driven fashion to learn how much to weigh state transition data (relative to choice data).

\subsubsection{Random discretization generator algorithm}
\label{ssec:random_partition_generator}
We now present the algorithm used for data generation. The intuition for this algorithm is very similar to recursive partitioning: in each step, we randomly select one of the partitions and split it into two partitions along one of the first 10 variables in $\Q$. Please note that we only used the first 10 variables in $\Q$ for partitioning, and the following 20 variables are added as irrelevant variables to show the robustness of the algorithm to irrelevant variables. In addition, to prevent very small or very large partitions, the random partition selection is weighted by the size of the partition: big partitions are more likely to be selected than smaller partitions. Formally, the random partitioning algorithm is as follows.

\begin{itemize}
    \item Initialize $\Pii_0$ as one partition equal to the full covariate space.
    \item Do the following for 15 rounds.
    \begin{itemize}
        \item Step 1: Randomly select a partition from $\Pii_r$ weighted by $(\frac{1}{\text{partition's total splits}})^2$
        \item Step 2: Randomly select a variable from the first 10 variables
        \item Step 3: Split the selected partition into two from the midpoint along the selected variable. If there is no exact midpoint to split, then pick the smaller of the two integers closest to the midpoint as the splitting point. Go back to Step 1 and repeat this round if the split is not possible\footnote{Recall that the variables in $\Q$ only take integer values from 0 to 10. Therefore, if the selected variable in Step 2, say $q_1$, starts at 9 and ends at 10 in the chosen partition, then further splits along $q_1$ within this partition are not possible.}.
    \end{itemize}
\end{itemize}

The total number of splits for each partition is the total of times their parents have selected to be split. If the total split for a partition is 5, it means that it takes five splits from $\Pii_0$ to get to this partition. Note that the total number of splits of a partition is negatively correlated with the size of the partition.

We also vary the replacement cost for each partition to add a decision variation to the model that is not caused by state transition. The replacement cost of a partition is calculated based on the range of its first 10 variables as the following.
\begin{equation}
    f_{dc}(\pi) = 5 - \frac{\sum_{i=1}^{10} (v_i^{min}(\pi) + v_i^{max}(\pi))}{10}
\end{equation}
where $v_i^{min}(\pi)$ and $v_i^{max}(\pi)$ are the minimum and maximum range of $i^{th}$ variable in partition $\pi$. We choose this formulation for the replacement cost to create a good balance between decision variation caused by state transition or replacement cost. Figure \ref{figure:partitionings_stat} depicts the distribution of total split and replacement costs in the 1500 generated partitions across all the 100 generated discretizations in the second simulation study.
\begin{figure}[!ht]
    \centering
    \includegraphics[width=16cm]{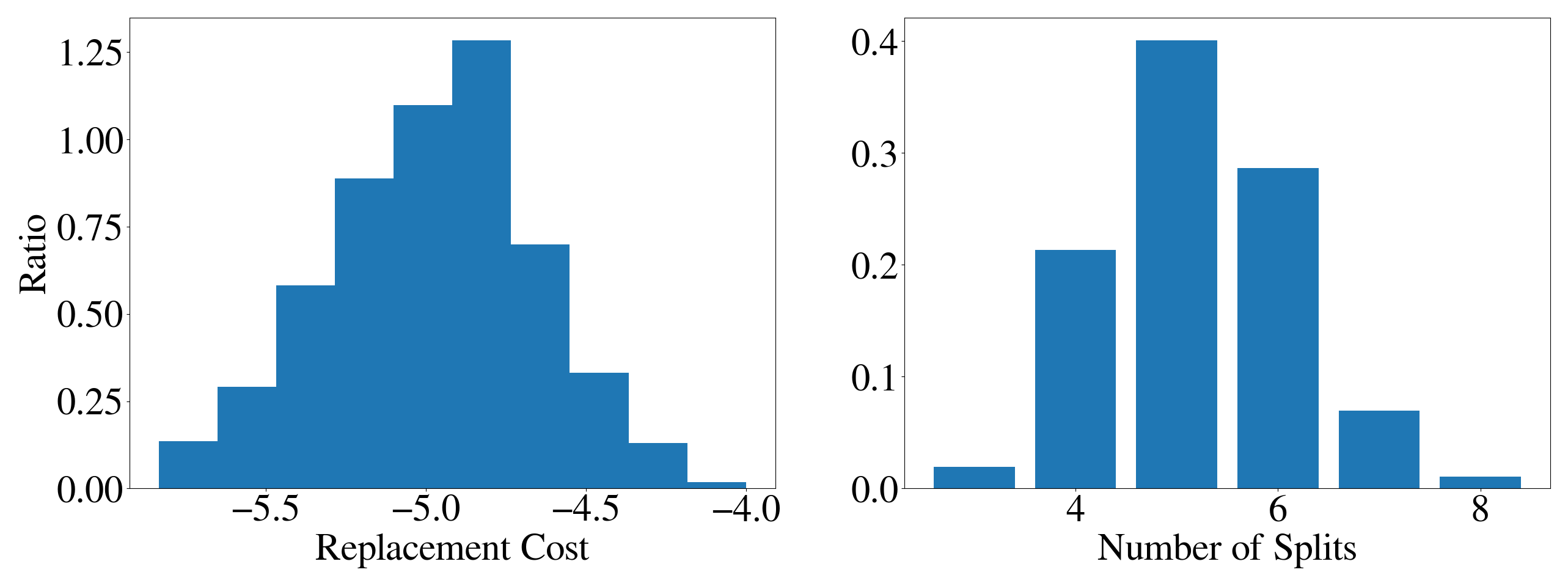}
    \caption{\label{figure:partitionings_stat} The histogram of generated partitions' calculated replacement cost, and number of splits in the 100 rounds of partitioning generation.}
\end{figure}.

\label{sec:random_partition_generator}

\subsection{Experiment 2: Comparison of Parameter Estimates}

We now consider a series of experiments to highlight how the choice of hyper-parameters can affect the quality of the parameter estimates. In particular, we focus on the role of the regularization parameter $\lambda_{rel}$, which captures the relative importance of decision-relevant and state-transition components of the objective function. 

\subsubsection{Data Generating Process}
For this exercise, we consider the second simulation study from $\S$\ref{ssec:finite_horizon_simulations} simulation focus on two contrasting cases—Case 5 and Case 7—that differ in how the high-dimensional state variables influence utility and state transitions. In Case 5, the high-dimensional state variables ($q_{it}$) affect the flow utility from purchase but do not influence the price transition. Also, the state transition in high-dimensional variables is fully random. As a result, the transition component of the objective function contains no useful variation for recovering the partitioning. In contrast, in Case 7, the high-dimensional state variables influence the transition dynamics (specifically, the price evolution) but not the flow utility. Moreover, the state transition in Case 7 is not entirely random: agents move to a random location within the same partition. This structure embeds useful information about the underlying partitioning in the state transitions, which our algorithm can potentially exploit when $\lambda_{rel}$ is large.

To assess the importance of hyper-parameter optimization in different data regimes, we consider four scenarios, where the number of consumers is set to 300, 1000, 3000, and 10,000, respectively. Varying the number of consumers (or observations) allows us to evaluate how estimation quality changes as sample size increases and whether correctly tuning $\lambda_{\text{rel}}$ becomes more critical in data-scarce or data-rich environments.

Finally, for each scenario, we consider three different values of $\lambda_{\text{rel}}$: 0, 1, and 10,000. Each simulation scenario is repeated 10 times, and we report the average results across these replications.

\subsubsection{Results}
The results from this exercise are shown in Table \ref{table:lambda_cases}. For each scenario, we show the results on two key outcomes: the number of irrelevant variables mistakenly selected for discretization and the estimated price coefficient (the true value is 0.5). 

\begin{table}[ht]
\centering

\begin{minipage}{0.48\textwidth}
\centering
\caption*{Case 5}
\begin{tabular}{cccc}
\toprule
Consumers & $\lambda_{\text{rel}}$ & $\hat{\alpha}$ & Wrong Var. \\
\midrule
\multirow{3}{*}{300}   & 0     & -0.44 & 0.80 \\
                       & 1     & -0.44 & 1.00 \\
                       & 10000 & -0.32 & 2.80 \\
\cmidrule(lr){1-4}
\multirow{3}{*}{1000}  & 0     & -0.48 & 0.90 \\
                       & 1     & -0.48 & 0.90 \\
                       & 10000 & -0.33 & 2.50 \\
\cmidrule(lr){1-4}
\multirow{3}{*}{3000}  & 0     & -0.50 & 0.90 \\
                       & 1     & -0.50 & 0.90 \\
                       & 10000 & -0.35 & 2.00 \\
\cmidrule(lr){1-4}
\multirow{3}{*}{10000} & 0     & -0.50 & 0.70 \\
                       & 1     & -0.50 & 0.70 \\
                       & 10000 & -0.36 & 1.80 \\
\bottomrule
\end{tabular}
\end{minipage}
\hfill
\begin{minipage}{0.48\textwidth}
\centering
\caption*{Case 7}
\begin{tabular}{cccc}
\toprule
Consumers & $\lambda_{\text{rel}}$ & $\hat{\alpha}$ & Wrong Var. \\
\midrule
\multirow{3}{*}{300}   & 0     & -0.39 & 3.00 \\
                       & 1     & -0.47 & 1.90 \\
                       & 10000 & -0.49 & 0.80 \\
\cmidrule(lr){1-4}
\multirow{3}{*}{1000}  & 0     & -0.42 & 2.50 \\
                       & 1     & -0.50 & 1.00 \\
                       & 10000 & -0.50 & 1.00 \\
\cmidrule(lr){1-4}
\multirow{3}{*}{3000}  & 0     & -0.41 & 2.80 \\
                       & 1     & -0.49 & 1.00 \\
                       & 10000 & -0.50 & 0.60 \\
\cmidrule(lr){1-4}
\multirow{3}{*}{10000} & 0     & -0.42 & 2.60 \\
                       & 1     & -0.50 & 0.90 \\
                       & 10000 & -0.50 & 0.00 \\
\bottomrule
\end{tabular}
\end{minipage}
\caption{\label{table:lambda_cases} 
Estimation results across different values of $\lambda_{\text{rel}}$ and consumer sample sizes for Case 5 (left) and Case 7 (right). Each cell reports the average of 10 simulation runs. $\hat{\alpha}$ denotes the estimated price coefficient, with the true value set to 0.5. ``Wrong Var.'' counts how many irrelevant variables were incorrectly selected for discretization. }
\end{table}

First, across both cases, we see that the quality of estimates suffers when the number of observations is small (e.g., 300 consumers). However, as the number of observations increases, the bias in the parameter estimates reduces (across all hyper-parameter values). This suggests that hyper-parameter optimization is particularly important when the amount of data (relative to the dimensionality of the setting) is small. 

Next, we discuss the results for each case in detail. In Case 5 (left panel), where the high-dimensional state space influences only the utility function (and hence the decision) and state transitions are fully random, the algorithm does well on both partition recovery and price coefficient estimation when $\lambda_{\text{rel}}$ is 0 or 1. This suggests that the decision-based variation has meaningful information for the discretization task. However, when $\lambda_{\text{rel}} = 10{,}000$, performance is quite poor, even in relatively data-rich settings (e.g., 10,000 customers). This is because the algorithm relies almost entirely on the state transition component of the objective function for discretization, which in this case contains no information about the true partitioning.

In contrast, in Case 7 (right panel), state transitions encode the latent partitioning structure, but the flow utility is not a function of the high-dimensional state.  In this case, the algorithm performs best when $\lambda_{\text{rel}}$ is 10{,}000. When $\lambda_{\text{rel}} = 0$, the algorithm focuses solely on decision-based variation and it struggles to recover the true partitioning and estimate the price coefficient, even with large sample sizes. Note that users' decisions always contain some information on state transition (through the expected future value functions). So, even when $\lambda_{\text{rel}} = 0$, there is some small amount of information on the underlying partition. But this information is relatively weak compared to directly using the state transition information in Case 7.

In summary, we find that the relative weight of the decision and transition components in the objective function should reflect the structure of the environment. As such, it is important to carefully tune $\lambda_{\text{rel}}$ to ensure the robust performance of the algorithm, especially in limited data settings. 



\section{Comparison with Other Approaches}
\label{appsec:comparisons}


We now present some simulations and comparisons to three other natural approaches. First, in $\S$\ref{subsec:RePaD_vs_CCP}, we explore two-step methods based on flexible/semi-parametric CCPs. Next, in $\S$\ref{subsec:RePaD_DimReduction}, we consider unsupervised state-space reduction methods -- $k$-means clustering and PCA. 


\subsection{Flexible Estimation of CCPs}
\label{subsec:RePaD_vs_CCP}

One potential approach to accommodating the high-dimensional state variables is to use the two-step estimation method \citep{Hotz_Miller_1993}. In this method, we first flexibly estimate CCPs and state transition using either flexible parametric or non-parametric approaches. Then, we use an analytical relationship between CCPs, state transitions, and value functions to derive value functions instead of using value function iteration. However, the two-step approach requires that the researcher has access to consistent and high-quality non-parametric estimates of CCPs and state-transitions at each combination of state variables. While CCPs can be flexibly estimated when the number of decisions is relatively small (using semi-parametric approaches such as neural networks), it is not possible to specify and estimate state transition models in high-dimensional settings. Indeed, in the reinforcement learning literature, this challenge is often referred to as the lack of knowledge of the ``model of world'' and it is well-known that model-based roll-out approaches like two-step estimators are not feasible (or error-prone) without oracle knowledge of state transitions in high-dimensional settings \citep{Zeng_etal_2023}. Since it is not possible to specify fully non-parametric state-transitions in the 10-dimensional $q$-space, we do not consider full parameter estimation with CCPs. 
Nevertheless, we present comparisons of CCPs based on the full high-dimensional set of high-variables vs. those based on the discretization. 

We consider the same set of 12 scenarios from our first simulation study in $\S$\ref{ssec:bus_engine} and consider two estimators: (1) a flexible logit with all the state variables and their first-order interactions and (2) a two-layer deep neural network. First, for both estimators, we consider the full set of state variables to predict replacement, i.e., both mileage ($x_{it}$) and the full set of high-dimensional state variables ($q_{it}$). Next, we consider the state space after discretization, i.e., mileage ($x_{it}$) and the lower-dimensional partitions ($\Pii(q_{it})$).

\begin{table}[htp!]
\begin{tabular}{c|ccc|cccc}
\toprule
\textbf{Case} & \textbf{Effect on}  & \textbf{Effect on} & \textbf{Transition} & \multicolumn{4}{c}{\multirow{1}{*}{Model and high-dimensional data type}}\\
\textbf{No.} & \textbf{Purchase}  & \textbf{Price} &    \textbf{in $\Pii^*(\Q)$}  &  Logistic & Neural net &  Logistic & Neural net\\
& \textbf{Utility}   & \textbf{Transition}   & \textbf{Space} & \multicolumn{2}{c}{Without discretizing $\Q$} & \multicolumn{2}{c}{After discretizing $\Q$}\\
\midrule
1 & Dissimilar & Dissimilar & No transition & \makecell{0.2843\\0.0021} & \makecell{0.2844\\0.0021} & \makecell{0.2796\\0.0007} & \makecell{0.2781\\0.0005} \\
2 & Dissimilar & Dissimilar & Random transition & \makecell{0.2925\\0.0096} & \makecell{0.2848\\0.0068} & \makecell{0.2682\\0.0012} & \makecell{0.2640\\0.0003} \\
3 & Dissimilar & Dissimilar & Sparse transition & \makecell{0.2204\\0.0113} & \makecell{0.2125\\0.0083} & \makecell{0.1932\\0.0017} & \makecell{0.1860\\0.0001} \\
4 & Dissimilar & Similar & No transition & \makecell{0.3417\\0.0019} & \makecell{0.3426\\0.0023} & \makecell{0.3364\\0.0005} & \makecell{0.3348\\0.0002} \\
5 & Dissimilar & Similar & Random transition & \makecell{0.3176\\0.0056} & \makecell{0.3158\\0.0049} & \makecell{0.3014\\0.0004} & \makecell{0.2999\\0.0001} \\
6 & Dissimilar & Similar & Sparse transition & \makecell{0.3145\\0.0052} & \makecell{0.3111\\0.0039} & \makecell{0.2999\\0.0007} & \makecell{0.2975\\0.0001} \\
7 & Similar & Dissimilar & No transition & \makecell{0.3158\\0.0013} & \makecell{0.3144\\0.0009} & \makecell{0.3139\\0.0006} & \makecell{0.3116\\0.0002} \\
8 & Similar & Dissimilar & Random transition & \makecell{0.3469\\0.0024} & \makecell{0.3462\\0.0023} & \makecell{0.3435\\0.0010} & \makecell{0.3397\\0.0001} \\
9 & Similar & Dissimilar & Sparse transition & \makecell{0.3054\\0.0039} & \makecell{0.3022\\0.0031} & \makecell{0.3005\\0.0016} & \makecell{0.2947\\0.0003} \\
10 & Similar & Similar & No transition & \makecell{0.3514\\0.0007} & \makecell{0.3503\\0.0004} & \makecell{0.3506\\0.0004} & \makecell{0.3491\\0.0002} \\
11 & Similar & Similar & Random transition & \makecell{0.3501\\0.0006} & \makecell{0.3530\\0.0014} & \makecell{0.3494\\0.0004} & \makecell{0.3478\\0.0001} \\
12 & Similar & Similar & Sparse transition & \makecell{0.3487\\0.0007} & \makecell{0.3475\\0.0003} & \makecell{0.3479\\0.0004} & \makecell{0.3462\\0.0000} \\
\bottomrule
\end{tabular}
\caption{\label{table:ccp_loss} Predictive performance of logistic regression and neural network models using the full set of high-dimensional ($\Q$) and using the discretized $\Q$. Each model is evaluated on a dataset with 40000 observations using both cross-entropy loss (top line) and mean squared error (MSE, bottom line). Cross-entropy loss is computed as $-\frac{1}{n} \sum_{i=1}^{n} \left[ p_i \log \hat{p}_i + (1 - p_i) \log (1 - \hat{p}_i) \right]$, and MSE is $\frac{1}{n} \sum_{i=1}^{n} (p_i - \hat{p}_i)^2$.}
\end{table}

The results from this exercise are presented in Table \ref{table:ccp_loss}. For each scenario, we show how the model-predicted choice probability compares to the real choice probability in the data. We consider two metrics of goodness of fit -- (1) cross-entropy loss and (2) MSE. We make a few observations based on the findings from this table. Firstly, we see that neural networks usually outperform logistic regression, but the difference is marginal. Secondly, and more importantly, we find that across all scenarios, predicted CCPs {\it after} discretization are closer to the true CCPs observed in the data (irrespective of the goodness-of-fit metric used). Partitioning is particularly helpful in cases where the high-dimensional state variables meaningfully affect state transitions and utilities (e.g., scenario 3). In summary, these findings suggest that even flexible estimators struggle to estimate choice probabilities in high-dimensional settings, in spite of the fact that the setting considered is relatively data-rich. Thus, partitioning before estimating CCPs can improve predictive accuracy. As such, we should view RePaD as a complementary tool for enhancing CCP-based methods in high-dimensional settings.

\subsection{Unsupervised Dimensionality Reduction Approaches}
\label{subsec:RePaD_DimReduction}

A key challenge in DDC estimation is the need to reduce the dimensionality of the state space while preserving decision-relevant information. We again consider the data from the first simulation and focus on scenario 2. The findings are similar for other scenarios as well, so we focus on just one scenario for expositional ease. 

We assess how well two naive dimensionality reduction techniques perform in this context: $k$-means clustering and Principal Component Analysis (PCA). 

\subsubsection{\texorpdfstring{$k$}-means clustering}
We first apply $k$-means clustering on the data and partition the state space into four clusters based on $q_{it}$. Note that allowing for four clusters reflects the structure of the data in the first simulation study, where the state space was divided into four meaningful regions. 

However, we find that the clustering results are poor and fail to meaningfully differentiate observations in a way that aligns with decision-making patterns. To quantify the quality of the clustering, we computed two commonly used clustering evaluation metrics: Adjusted Rand Index (ARI) = 0.1654 and Normalized Mutual Information (NMI) = 0.2491. Both values indicate poor clustering performance, suggesting that the $k$-means algorithm failed to recover the underlying decision-relevant partitions. 

\subsubsection{PCA}
Next, we apply PCA on the high-dimensional state variables $q_{it}$ and reduce the dimensionality of this set of state variables from 10 to 2. In Figure \ref{figure:pca}, we visualize how well the two principal components captured the true partitions in our data. As expected, the PCA-transformed dimensions are unable to differentiate between the true partitions. This demonstrates that variance-based low-dimensional projections alone are insufficient for preserving the structure necessary to capture the relevant state-space for the agents' dynamic decisions. 

\begin{figure}[htp!]
    \centering
    \includegraphics[width=150mm]{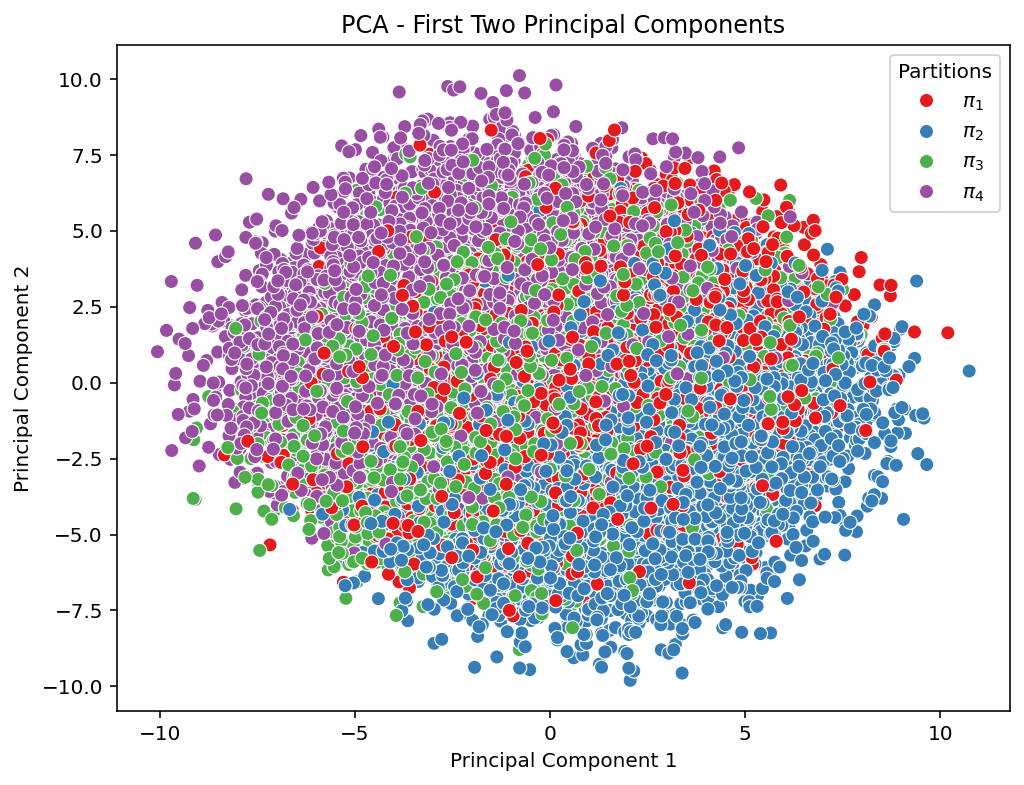}
    \caption{\label{figure:pca} The plot shows the projection of the high-dimensional state space onto the first two principal components obtained via PCA. The original partitions ($\pi_1, \pi_2, \pi_3, \pi_4$) are widely dispersed across the transformed space, indicating that PCA fails to recover the underlying decision and state transition-relevant structure. } 
\end{figure}


Intuitively, since both PCA and $k$-means clustering are unsupervised dimensionality reduction techniques, they fail to accurately capture the variation in state variables that affect state transitions and decisions. Instead, they simply capture the correlations in $q_{it}$, and in many cases, this correlation may be irrelevant to the agents' decisions and state transitions.

\end{appendices}
\end{document}